\PassOptionsToPackage{dvipsnames}{xcolor}
\documentclass[lettersize,journal]{IEEEtran}
\usepackage{amsmath,amsfonts}

\usepackage{algorithmicx,algorithm}
\usepackage{array}
\usepackage[caption=false,font=normalsize,labelfont=sf,textfont=sf]{subfig}
\usepackage{textcomp}
\usepackage{stfloats}
\usepackage{url}
\usepackage{verbatim}
\usepackage{cite}
\usepackage{algpseudocode}
\usepackage{bm}
\usepackage{multirow}
\usepackage{amsthm}
\usepackage{tcolorbox}
\usepackage{wrapfig}
\usepackage{pifont}
\usepackage{wrapfig,epsfig}
\usepackage{xcolor}
\usepackage[
color=blue!20!white,
textsize=tiny
]{todonotes}

\usepackage{graphicx}
\usepackage{authoraftertitle}

\usepackage{grffile}

\usepackage{balance}

\usepackage[normalem]{ulem}

\newtheorem{Definition}{Definition}
\newtheorem{Lemma}{Lemma}
\newtheorem{Assumption}{Assumption}
\newtheorem{Remark}{Remark}
\newtheorem{Corollary}{Corollary}
\newtheorem{Theorem}{Theorem}

\newcommand{\blue}[1]{{\color{black}#1}}
\newcommand{\green}[1]{{\color{black}#1}}

\newcommand{\AlgorithmFullName}{Byzantine-resilient Decentralized SGD~}

\newcommand{\argmin}[1]{\underset{#1}{\arg \min }}
\newcommand{\argmax}[1]{\underset{#1}{\arg \max }}

\newcommand{\1}{\bm{1}}

\newcommand{\vp}{{\ensuremath{\boldsymbol p}}}
\newcommand{\vx}{{\ensuremath{\boldsymbol x}}}
\newcommand{\vy}{{\ensuremath{\boldsymbol y}}}
\newcommand{\vz}{{\ensuremath{\boldsymbol z}}}
\newcommand{\vxi}{{\ensuremath{\boldsymbol \xi}}}
\newcommand{\N}{\mathcal{N}}
\newcommand{\B}{\mathcal{B}}
\newcommand{\edge}{\mathcal{E}}

\newcommand{\cD}{\mathcal{D}}

\newcommand{\cU}{\mathcal{U}}

\def\ccalW{{\ensuremath{\mathcal W}}}
\newcommand{\A}{\mathcal{A}}

\newcommand{\R}{\mathbb{R}}
\newcommand{\E}{\mathbb{E}}

\newcommand{\lp}{\left(}
\newcommand{\rp}{\right)}

\newcommand{\lb}{\left\{}
\newcommand{\rb}{\right\}}
\newcommand{\la}{\left\langle}
\newcommand{\ra}{\right\rangle}
\newcommand{\lnorm}{\left\|}
\newcommand{\rnorm}{\right\|}

\makeatletter
\newsavebox\myboxA
\newsavebox\myboxB
\newlength\mylenA
\newcommand*\xbar[2][0.75]{%
	\sbox{\myboxA}{$\m@th#2$}%
	\setbox\myboxB\null
	\ht\myboxB=\ht\myboxA%
	\dp\myboxB=\dp\myboxA%
	\wd\myboxB=#1\wd\myboxA
	\sbox\myboxB{$\m@th\overline{\copy\myboxB}$}
	\setlength\mylenA{\the\wd\myboxA}
	\addtolength\mylenA{-\the\wd\myboxB}%
	\ifdim\wd\myboxB<\wd\myboxA%
	\rlap{\hskip 0.5\mylenA\usebox\myboxB}{\usebox\myboxA}%
	\else
	\hskip -0.5\mylenA\rlap{\usebox\myboxA}{\hskip 0.5\mylenA\usebox\myboxB}%
	\fi}
\makeatother

\newcommand{\smallCE}[1]{$<$1e-07}

\newcommand{\ifinarxiv}[2]{\ifthenelse{\isundefined{\inarxiv}}{#2}{#1}}
\usepackage{ifthen}
\def\inarxiv{1}

\begin{document}

	\title{Byzantine-Resilient Decentralized Stochastic Optimization with Robust Aggregation Rules}
	
	\author{Zhaoxian Wu\IEEEauthorrefmark{1},
		Tianyi Chen\IEEEauthorrefmark{2}, and
		Qing Ling\IEEEauthorrefmark{1}\IEEEauthorrefmark{3}
		\thanks{
			\IEEEauthorrefmark{1}School of Computer Science and Engineering and Guangdong Provincial Key Laboratory of Computational Science, Sun Yat-Sen University, Guangzhou, Guangdong 510006, China.
			\IEEEauthorrefmark{2}Department of Electrical, Computer, and Systems Engineering, Rensselaer Polytechnic Institute, Troy, New York 12180, USA.
			\IEEEauthorrefmark{3}Pazhou Lab, Guangzhou, Guangdong 510330, China.
			Zhaoxian Wu and Qing Ling are supported by National Natural Science Foundation of China under grant 61973324, Guangdong Basic and Applied Basic Research Foundation under grant 2021B1515020094, and Guangdong Provincial Key Laboratory of Computational Science at Sun Yat-Sen University under grant 2020B1212060032.
		}
	}
	\ifinarxiv{
	}{
		\markboth{IEEE Transactions on Signal Processing, Submitted}%
		{Shell \MakeLoweRCAe{\textit{et al.}}: A Sample Article Using IEEEtran.cls for IEEE Journals}
		
	}
	
	\maketitle
	\begin{abstract}
		This paper focuses on decentralized stochastic optimization in the presence of Byzantine attacks. During the optimization process, an unknown number of malfunctioning or malicious workers, termed as Byzantine workers, disobey the algorithmic protocol and send arbitrarily wrong messages to their neighbors. Even though various Byzantine-resilient algorithms have been developed for distributed stochastic optimization with a central server, we show that there are two major issues in the existing robust aggregation rules when being applied to the decentralized scenario: disagreement and non-doubly stochastic virtual mixing matrix. This paper provides comprehensive analysis that discloses the negative effects of these two issues, and gives guidelines of designing favorable Byzantine-resilient decentralized stochastic optimization algorithms. Under these guidelines, we propose iterative outlier scissor (IOS), an iterative filtering-based robust aggregation rule with provable performance guarantees. Numerical experiments demonstrate the effectiveness of IOS. The code of simulation implementation is available at \url{github.com/Zhaoxian-Wu/IOS}.
	\end{abstract}
	
	\begin{IEEEkeywords}
		Decentralized network, stochastic optimization, Byzantine attacks, robust aggregation rule.
	\end{IEEEkeywords}
	
	\vspace{-1em}
	
	\section{Introduction}
	\label{section:introduction}
	Training large machine learning models  relies on vast amounts of data to achieve accurate predictions.
	However, data is often dispersed among geographically distributed devices, or workers, and is subject to growing privacy concerns. To address this issue, distributed or decentralized stochastic optimization has been proposed as a means of  privacy-preserving model training \cite{Konecn2016FederatedOD, kairouz2019advances, che2021decentralized}. Distributed stochastic optimization involves a central server that exchanges messages, such as stochastic gradients or intermediate models. In contrast, in the decentralized scenario, workers exchange messages in a peer-to-peer manner. Decentralized topology offers better scalability and avoids communication bottlenecks that can occur in the distributed counterpart, as there is no dependence on a central server.

	However, distributed or decentralized stochastic optimization faces potential robustness issues due to the involvement of a vast number of workers. Data corruption, device malfunctioning, or malicious attacks can cause some workers to deviate from the training protocol, which we refer to as \emph{Byzantine attacks}. The abnormal workers, called as \emph{Byzantine workers}, are assumed to be omniscient and arbitrarily malicious, while their number and identities are unknown to the central server and the \emph{honest workers} \cite{lamport1982, yang2020adversary}.
	Although Byzantine-resilient distributed stochastic optimization has been extensively studied in literature, much less attention has been paid to the decentralized scenario. The aim of this paper is to highlight challenges of Byzantine-resilient decentralized stochastic optimization, provide algorithmic guidelines, and propose an effective approach.
	Below we briefly review the literature.
	
	
	\textbf{Byzantine-resilient distributed stochastic optimization.} To defend against Byzantine attacks in distributed stochastic optimization, most existing algorithms substitute the vulnerable mean aggregation rule in distributed stochastic gradient descent (SGD) with robust aggregation rules. Examples of such aggregation rules include coordinate-wise median \cite{yin2018ByzantineRobustDL}, geometric median \cite{chen2019DistributedSM, xie2018GeneralizedBS}, trimmed mean \cite{xie2018phocas}, Krum \cite{blanchard2017MachineLW}, Bulyan \cite{guerraoui2018hidden}, FABA \cite{xia2019faba}, centered clipping (CC) \cite{gorbunov2021secure}. The key idea of these robust aggregation rules is to find a point that is close to the mean of the stochastic gradients transmitted from the honest workers. When the stochastic gradients from the honest workers are i.i.d. (independent and identically distributed) but subject to large noise, finding such a point is difficult. This fact motivates the applications of variance reduction techniques \cite{wu2020federated,khanduri2019byzantine} and momentum \cite{karimireddy2021learning} to alleviate the stochastic gradient noise, and consequently, boost the Byzantine-resilience. Although these methods have been shown effective in the distributed scenario, directly extending them to the decentralized scenario does not yield similar results, as we will discuss in Section \ref{section:challenge-of-robust-aggregations}.
	
	When the data across the honest workers is non-i.i.d., the stochastic gradients from the honest workers are non-i.i.d. as well,  making approximating their mean more challenging \cite{data2021byzantine}. To overcome this heterogeneity issue, several techniques have been developed, such as robust stochastic model aggregation (RSA) \cite{li2019RSABS} and resampling/bucketing techniques \cite{karimireddy2022byzantine, peng2022byzantine}.
	
	Besides those based on robust aggregation rules, some other algorithms have been proposed to identify Byzantine workers during the distributed training process \cite{xie2018zeno,li2020learning}, followed by eliminating their influence.

	%
	
	\textbf{Byzantine-resilient decentralized (stochastic) optimization.} One of the most popular decentralized stochastic optimization algorithms is decentralized SGD \cite{nedic2009distributed,lian2017can}. Different to distributed SGD in which the central server uses mean to aggregate stochastic gradients from all workers, in decentralized SGD, each worker uses weighted mean to aggregate optimization variables (models) from its neighbors, followed by a local SGD step. However, decentralized SGD fails even when one Byzantine worker exists. The Byzantine worker can disturb the training processes of its honest neighbors, and further affect those of all honest workers across the entire network, through the diffusion of polluted messages.

	Based on trimmed mean aggregation, the works of \cite{su2015fault,su2020byzantine,fang2022bridge,yang2019byrdie} study Byzantine-resilient decentralized {\em deterministic} optimization. 
	Byzantine-resilient decentralized {\em stochastic} optimization is relatively less investigated. In \cite{el2021collaborative}, the equivalence between Byzantine-resilient agreement and stochastic optimization is highlighted, but the investigated network topology is confined to be complete. The work of \cite{peng2021byzantine} extends RSA from the distributed scenario to decentralized, and \cite{he2022byzantine} extends distributed CC to decentralized self centered clipping (SCC). They require to set task-dependent hyper-parameters, which are hard to tune.
	The work of \cite{guo2021byzantine} proposes a two-stage method to filter Byzantine attacks, and \cite{xu2022byzantine} proposes a similarity-based method to aggregate neighboring messages. In addition, trimmed mean can be implemented in the stochastic scenario too. However, most of these methods fail to display favorable Byzantine-resilience under certain attacks as we will show with numerical experiments. More importantly, principled guidelines to develop effective \emph{Byzantine-resilient decentralized stochastic optimization} algorithms are still lacking.
	
	
	
	The contributions of this paper are summarized as follows.
	\begin{itemize}
		
		\item [\bf C1)] We show a wide class of existing robust aggregation rules in the distributed scenario fail to reach consensus in the decentralized scenario even when no Byzantine workers are present. In addition, many of them lead to non-doubly stochastic virtual mixing matrices (see Definition \ref{definition:mixing-matrix}). We theoretically demonstrate that both issues enlarge the asymptotic learning error of a Byzantine-resilient decentralized stochastic optimization algorithm.

		\item [\bf C2)] Leveraging the theoretical analysis, we provide guidelines to design a Byzantine-resilient decentralized stochastic optimization algorithm; that is
		the robust aggregation rules of honest workers should satisfy the following criteria: a small contraction constant (see Definition \ref{definition:mixing-matrix}) and a doubly stochastic virtual mixing matrix.
		
		\item [\bf C3)] Following these design guidelines, we propose a novel robust aggregation rule for decentralized stochastic optimization, termed as iterative outlier scissor (IOS), and validate its superior performance via extensive numerical experiments.
		
	\end{itemize}
	
	
	\vspace{-1em}
	\section{Byzantine-Resilient Decentralized\\ Stochastic Optimization}
	\label{sec:dso}
	
	Consider an undirected graph $\mathcal{G} := (\N\cup\B, \mathcal{E})$, where $\N$ and $\B$ respectively denote the sets of honest and Byzantine workers, and $\mathcal{E} \subseteq (\N\cup\B)\times(\N\cup\B)$ denotes the set of edges without self-links. Note that the number and identities of the Byzantine workers are unknown to the honest workers. When an edge $(n, m)\in\edge$ exists, workers $n$ and $m$ are neighbors and can communicate with each other. For any worker $n$, denote the sets of its honest and Byzantine neighbors as $\N_n:=\{m | (m, n)\in\edge, m\in\N\}$ and $\B_n:=\{m | (m, n)\in\edge, m\in\B\}$, respectively.
	Define the numbers of honest and Byzantine workers as $N:=|\N|$ and $B:=|\B|$, respectively. For any worker $n$, define the numbers of its honest and Byzantine neighbors as $N_n=|\N_n|$ and $B_n=|\mathcal{B}_n|$, respectively.
	
	
	With these notations, the problem of \textit{Byzantine-resilient decentralized stochastic optimization} can be described as finding an optimal solution of the following problem
	\begin{align}
		\label{problem}
		& \vx^* \in \argmin{\vx \in\R^D}~ f(\vx) := \frac{1}{N} \sum_{n\in\N} f_n(\vx), \\
		& \hspace{2.7em} \text{where} ~ f_n(\vx) := \E_{\xi_n} [f_n(\vx; \xi_n)]. \notag
	\end{align}
	In \eqref{problem}, $f_n(\vx)$ is the local aggregated cost function of honest worker $n\in\N$ with optimization variable $\vx \in\R^D$; $f_n(\vx; \xi_n)$ is the local cost function associated with random variable $\xi_n$, which follows the local distribution $\cD_n$. The local distributions across the honest workers can be different.
	For notational convenience, we collect the random variables into $\vxi:= [\xi_n]_{n\in\N}$ and define the overall cost function
	\begin{align}
		\label{objective}
		f(\vx;\vxi) := \frac{1}{N} \sum_{n\in\N} f_n(\vx; \xi_n),
	\end{align}
	such that \eqref{problem} amounts to minimizing $\E_{\vxi} [f(\vx;\vxi)]$.
	
	
	
	In general, a decentralized algorithm to solve \eqref{problem} contains three stages: computation, communication, and aggregation. Next we show the implementation of popular decentralized SGD algorithms \cite{nedic2009distributed,lian2017can} when Byzantine workers are present. At time $k$, each honest worker $n \in \N$ independently samples a random variable $\xi_n^k \sim \cD_n$, computes the stochastic gradient $\nabla f_n(\vx^{k}_n; \xi_n^{k})$ using its current variable (also termed as model) $\vx_n^{k}$, and updates an intermediate variable $\vx_n^{k+\frac{1}{2}}$ by
	\begin{align}
		\label{rule:gradient-descent}
		\vx_n^{k+\frac{1}{2}} = \vx_n^{k} - \alpha^{k} \nabla f_n(\vx^{k}_n; \xi_n^{k}),
	\end{align}
	where $\alpha^k > 0$ is the step size.
	
	When Byzantine workers are absent, worker $n$ sends $\vx_n^{k+\frac{1}{2}}$ to and receives $\vx_m^{k+\frac{1}{2}}$ from all neighbors, followed by weighted mean aggregation. However, in the presence of Byzantine attacks, each honest worker $n \in \N$ cannot distinguish honest neighbors $m \in \N_n$ and Byzantine neighbors $m \in \B_n$. If $m\in\B$, it can transmit an arbitrary vector to its  neighbors. Byzantine worker $m\in\B$ may send different messages to different honest neighbors $n \in \N_m$. Thus, let $\tilde{\vx}_{m, n}^{k+\frac{1}{2}}$ denote the message that worker $m$ sends to worker $n$ at time $k$, as
	\begin{align}
		\label{equation:gk}
		\tilde{\vx}_{m, n}^{k+\frac{1}{2}} :=
		\begin{cases}
			\vx_m^{k+\frac{1}{2}}, \quad & m\in\N, \\
			*, \quad & m\in\B,
		\end{cases}
	\end{align}
	where $*$ denotes an arbitrary vector in $\R^D$. After that, each honest worker $n \in \N$ takes a  weighted average of the received messages to update its model $\vx_n^{k+1}$ by
	\begin{align}
		\label{rule:DECENTRALIZED SGD}
		\vx_n^{k+1} = \sum_{m\in\N_n\cup\B_n\cup\{n\}} w^{\prime}_{nm} \tilde{\vx}_{m,n}^{k+\frac{1}{2}}.
	\end{align}
	Here, $w^{\prime}_{nm} \geq 0$ is the weight that honest worker $n$ assigns to its neighbor (or itself) $m$, with $\sum_{m\in\N_n\cup\B_n\cup\{n\}} w^{\prime}_{nm} = 1$.
	Such an aggregation rule is vulnerable to Byzantine attacks. A Byzantine neighbor $m \in \B_n$ can arbitrarily manipulate $\vx_n^{k+1}$, for example, making $\vx_n^{k+1}=0$ by nullifying the weighted average in \eqref{rule:DECENTRALIZED SGD} or blowing $\vx_n^{k+1}$ up by sending a message with infinitely large elements. Even worse, using $n$ as an intermediate, $m$ can indirectly affect the honest neighbors of $n$ throughout the information diffusion process \cite{peng2021byzantine}.
	
	%
	
	\begin{algorithm}[t]
		\caption{\AlgorithmFullName}
		\label{algorithm:ByrdDec}
		\begin{algorithmic}[1]
			\Require step size $\alpha^k$; initialization $\vx^0_n=\vx^0$ for all $n \in \N$
			\ForAll {$k = 0, 1, 2, \cdots$}
			\ForAll {honest workers $n\in\N$}
			\State Compute stochastic gradient $\nabla f_n(\vx^{k}_n; \xi_n^{k})$
			\State Compute $\vx_n^{k+\frac{1}{2}} = \vx^k_n - \alpha^{k} \nabla f_n(\vx^{k}_n; \xi_n^{k})$
			\State Send $\tilde{\vx}^{k}_{n, m}=\vx_n^{k+\frac{1}{2}}$ to all neighbors $m$
			\State Receive $\tilde{\vx}_{m,n}^{k}$ from all neighbors $m$
			\State Aggregate $\vx^{k+1}_n=\A_n (\vx_n^{k+\frac{1}{2}}, \{\tilde{\vx}_{m,n}^{k+\frac{1}{2}}\}_{m\in\N_n\cup\B_n})$
			\EndFor
			\ForAll {Byzantine workers $n\in\B$}
			\State Send $\tilde{\vx}^{k}_{n, m}=*$ to all neighbors $m$
			\EndFor
			\EndFor
		\end{algorithmic}
	\end{algorithm}

	A remedy to address this issue is replacing the non-robust weighted average in \eqref{rule:DECENTRALIZED SGD} with an aggregation rule that is robust to Byzantine attacks. For each honest worker $n \in \N$, define its robust aggregation rule
	$\A_n: \R^D\times\R^{(N_n+B_n)\times D}\to\R^{D}$ as
	\begin{align}
		\label{rule:SGD-robust}
		\vx_n^{k+1} = \A_n (\vx_n^{k+\frac{1}{2}}, \{\tilde{\vx}_{m,n}^{k+\frac{1}{2}}\}_{m\in\N_n\cup\B_n}).
	\end{align}
	Thus, we modify decentralized SGD to its Byzantine-resilient variant, outlined in Algorithm \ref{algorithm:ByrdDec}.
	Specific choices of the robust aggregation rule will be discussed in the next section.

	\vspace{-0.8em}
	
	\section{Robust Aggregation Rules and Our Proposal}
	\label{sec:rar}
	
	This section starts with a generic form of robust aggregation rules in existing Byzantine-resilient distributed stochastic optimization algorithms. Then we show empirically that direct extensions of some robust aggregation rules to the decentralized scenario may fail. Based on this, we introduce our robust aggregation rule termed as \emph{iterative outlier scissor} (IOS).
	
	\vspace{-0.8em}
	
	\subsection{Generic Form of Robust Aggregation Rules}
	
	We consider the following generic form of the robust aggre- gation rule $\A_n$ for an honest worker $n \in \N$, given by
	
	{\centering
		\begin{tcolorbox}[width=0.49\textwidth]
			\vspace{-1.2em}
			\begin{align}
				\label{definition:aggregation-centralized-to-decentralized}
				\vspace{-0.5em}
				\!\!\!\!\!&       \A_n  (\vx_n, \{\tilde{\vx}_{m,n}\}_{m\in\N_n\cup\B_n})  \!\!\! \\
				\!\!\!\!\! &~~
				:=  (1-r_n)  \A(\vx_n, \{\tilde{\vx}_{m,n}\}_{m\in\N_n\cup\B_n}) + r_n \vx_n, \nonumber
			\end{align}
			\vspace{-1.8em}
	\end{tcolorbox}}
	\hspace{-0.35cm}where $\A$ is a \textit{base aggregator} that is common among all honest workers and $r_n \in [0, 1)$ is a worker-specific constant. Given all the available messages, the base aggregator outputs a $D$-dimensional vector. The output of the base aggregator might lose the information of $\vx_n$, in which honest worker $n$ should trust. Therefore, we consider the convex combination of the output of the base aggregator and $\vx_n$, parameterized by $r_n$. A small $r_n$ means that $n$ trusts more in the base aggregator.

	We will show in Appendix \ref{section:coverage-of-generic-form} that the existing Byzantine-resilient decentralized (stochastic) optimization algorithms all fall in the generic form of \eqref{definition:aggregation-centralized-to-decentralized}.
	Besides, this generic form also allows us to extend various base aggregators to Byzantine-resilient decentralized stochastic optimization.
	
	Below we discuss a number of popular base aggregators in Byzantine-resilient stochastic optimization as examples.
	For notational convenience, in these examples we define the input   of $\A$ as $\vx_1, \ldots, \vx_S$, where $S=N_n+B_n+1$.
	
	\textit{Coordinate-wise median (CooMed)}  returns the median for each coordinate $d=1, \ldots, D$ as \cite{yin2018ByzantineRobustDL}
	\begin{align}
		\left[\mathrm{CooMed}\left(\vx_{1}, \ldots, \vx_{S}\right)\right]_{d}=\operatorname{Median}\left(\left[\vx_{1}\right]_{d}, \ldots,\left[\vx_{S}\right]_{d}\right),
	\end{align}
	in which $[\vx]_d$ refers to the $d$-th element of vector $\vx$.
	
	\textit{Geometric median (GeoMed)}  finds a point that minimizes the sum of distances to all input vectors, given by \cite{chen2019DistributedSM, xie2018GeneralizedBS}
	\begin{align}
		\operatorname{GeoMed}\left(\vx_{1}, \ldots, \vx_{S}\right)=\underset{\vx}{\arg \min } \sum_{n=1}^{S}\left\|\vx-\vx_{n}\right\|.
	\end{align}
	
	
	In addition, if we can estimate the number of Byzantine workers (or a reasonable upper bound), denoted as $q$, we can also apply \textit{Krum} \cite{blanchard2017MachineLW}, which returns the input vector that has the minimal distance to $S-q-1$ nearest vectors, given by
	\begin{align}
		& \operatorname{Krum}\left(\vx_{1}, \ldots, \vx_{S}\right) \\
		=& \underset{\boldsymbol{x}_{n},n\in\{1\dots,S\}}{\arg \min } \min _{\cU: \cU \subset [S],\atop |\cU|=S-q-2} \sum_{m \in \mathcal{\cU}}\left\|\vx_n-\vx_m\right\|^{2}. \nonumber
	\end{align}

	
	However, below we empirically show the failure cases of these robust aggregation rules when being applied to the decentralized scenario; see the simulation results in Fig. \ref{fig:covergence} and more details in Section \ref{section:experiment} and Appendix \ref{section:supplementary-experiments}. Consider an Erdos-Renyi graph with $10$ honest workers and $2$ Byzantine workers. Each pair of workers are neighbors with probability of $0.7$.
	Surprisingly, though CooMed, GeoMed and Krum have been proven successful in Byzantine-resilient distributed stochastic optimization, they perform poorly in the decentralized scenario. In contrast, our proposed IOS algorithm archives the highest accuracy and low disagreement measure. Next we introduce IOS first, explain why CooMed, GeoMed and Krum fail in Section \ref{section:challenge-of-robust-aggregations}, and reveal the generic design guidelines in Sections \ref{section:convergence-learning-error} and \ref{section:algorithm_design}.
	
	\begin{figure}[t]
		\centerline{\includegraphics[width=1\columnwidth]{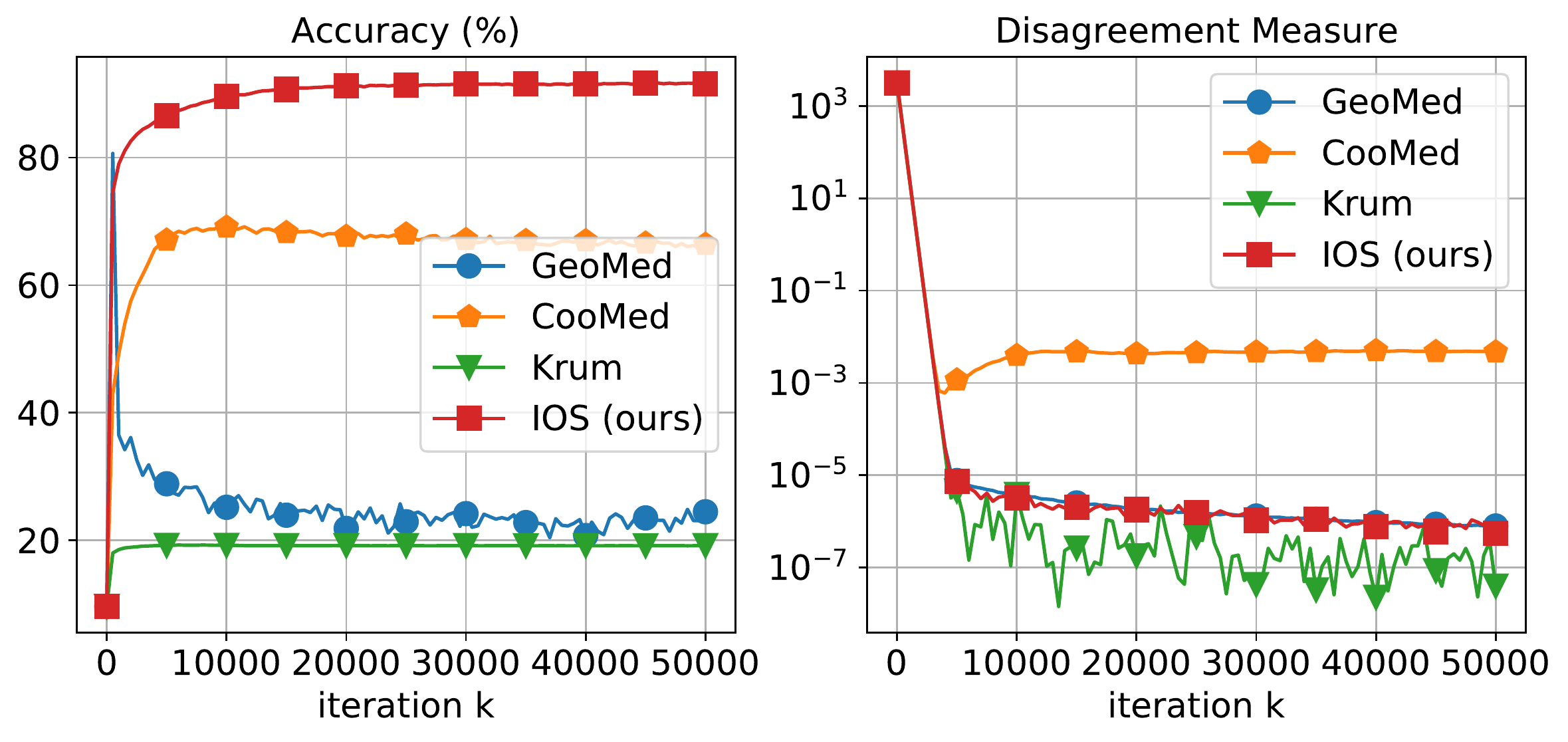}}
		\caption{Comparison in the Erdos-Renyi graph under sign-flipping attacks for the non-i.i.d. case. Performance metrics are average classification accuracy and disagreement measure, both in terms of local models of honest workers.}
		\label{fig:covergence}
	\end{figure}

	\vspace{-1em}
	
	\subsection{Our Proposal: Iterative Outlier Scissor}
	We propose a novel robust aggregation rule,  \textit{iterative outlier scissor (IOS)}, which iteratively discards outliers.
	IOS begins from constructing a doubly stochastic and symmetric mixing matrix $W'\in\R^{(N+B)\times(N+B)}$ in a decentralized manner with the existing techniques, such as the Metropolis-Hastings rule\footnote{In many existing decentralized approaches to constructing $W'$, each worker needs to know the degrees of all neighbors. Therefore, Byzantine neighbors can report wrong degrees in purpose. To address this issue, we can run Byzantine-resilient topology discovery algorithms (see \cite{nesterenko2009discovering} for an example) in advance and then construct $W'$.} \cite{xiao2004fast}.
	For notational convenience, define the cumulative weight of worker $n$ with respect to set $\cU$ as
	\begin{align}
		\label{eq:wprime}
		{\ccalW'_n}(\cU) := \sum_{m\in\cU} w_{nm}', ~~~~\cU\subseteq\N_n\cup\B_n\cup\{n\}.
	\end{align}
	
	At each time $k$, each honest worker $n \in \N$ receives messages from the workers in $\N_n\cup\B_n\cup\{n\}$. With IOS, it iteratively discards $q_n$ messages, where $q_n$ is the  estimated number of its Byzantine neighbors.
	At inner iteration $i$, each honest worker $n \in \N$ maintains a trusted set $\cU_n^{(i)}$, initialized as $\N_n\cup\B_n\cup\{n\}$ when $i=0$. It computes the weighted average of all models in $\cU_n^{(i)}$, denoted as $\vx^{(i)}_{\text{avg}}$, with \eqref{definition:ios-x-avg}. Then, it discards the model that is farthest away from $\vx^{(i)}_{\text{avg}}$ except its own model  in \eqref{definition:ios-m-max} and accordingly modifies the trusted set to $\cU_n^{(i+1)}$.
	This process repeats until $q_n$ models have been discarded, and outputs the weighted average of trusted models, denoted as $\vx^{(q_n)}_{\text{avg}}$, with \eqref{definition:ios-x-avg-qn}; see the summary in Algorithm \ref{algorithm:ios}.
	
	Under the notation of the generic robust aggregation rule in \eqref{definition:aggregation-centralized-to-decentralized}, the IOS aggregation at each honest worker $n \in \N$ is
	
	{\centering
		\begin{tcolorbox}[width=0.49\textwidth]
			\vspace{-1.2em}
			\begin{align}
				\label{rule:IOS-basic-aggregator-form}
				&\A_n (\vx_n, \{\tilde{\vx}_{m,n}\}_{m\in\N_n\cup\B_n})
				\nonumber \\
				&~~~~~= (1-r_n) \!\!\sum_{m\in \cU^{(q_n)}_n \setminus \{n\}} \!\!w''_{nm}\tilde{\vx}_{m,n} + r_n \vx_n,
			\end{align}
			\vspace{-1.2em}
	\end{tcolorbox}}
	
	\noindent where $r_n=w_{nn}'/{\ccalW'_n}(\cU^{(q_n)}_n)$ and $w''_{nm} := w_{nm}'/({\ccalW'_n}(\cU^{(q_n)})$ $-w'_{nn})$, respectively.
	
	IOS is inspired by FABA \cite{xia2019faba}, which iteratively discards models that are farthest away from the average model, not the weighted average in IOS.
	However, FABA was originally designed for the distributed scenario; extending it to decentralized often  leads to a non-doubly stochastic virtual mixing matrix (see Definition \ref{definition:mixing-matrix}), and consequently, comes with a large asymptotic learning error. We will show the limitation of the decentralized extension of FABA in Section \ref{section:experiment}.

	\begin{algorithm}[t]
		\caption{Iterative outlier scissor on honest worker $n$}
		\label{algorithm:ios}
		\begin{algorithmic}[1]
			\Require models $\{\vx_n\}\cup\{\tilde{\vx}_{m, n}\}_{m\in\N_n\cup\B_n}$;
			\Statex \hspace{2.2em} weights $\{w_{nm}'\}_{m\in\N_n\cup\B_n\cup\{n\}}$;
			\Statex \hspace{2.2em} estimate of number of Byzantine neighbors $q_n$
			\State Construct initial trusted set $\mathcal{U}_n^{(0)}=\N_n\cup\B_n\cup\{n\}$
			\For{$i = 0, 1, \cdots, q_n-1$}
			\State Compute weighted average of models
			\begin{align}
				\label{definition:ios-x-avg}
				\vx^{(i)}_{\text{avg}}=\frac{1}{\green{\ccalW'_n}(\cU_n^{(i)})}\sum_{m\in \cU_n^{(i)}}w_{nm}'\tilde{\vx}_{m,n}
			\end{align}
			\State Choose index
			\begin{align}\label{definition:ios-m-max}
				m^{(i)} = \argmax{m\in\mathcal{U}_n^{(i)}\setminus \{n\}} \|\tilde{\vx}_{m,n}-\vx^{(i)}_{\text{avg}}\|
			\end{align}
			\State Discard $m^{(i)}$ from trusted set
			\begin{align}
				\mathcal{U}_n^{(i+1)} = \mathcal{U}_n^{(i)} \setminus \{m^{(i)}\}
			\end{align}
			\EndFor
			\State Compute weighted average of trusted models
			\begin{align}
				\label{definition:ios-x-avg-qn}
				\vx^{(q_n)}_{\text{avg}}=
				\frac{1}{\green{\ccalW'_n}(\cU_n^{(q_n)})}\sum_{m\in \cU_n^{(q_n)}}w_{nm}'\tilde{\vx}_{m,n}
			\end{align}
			\State \Return $\vx^{(q_n)}_{\text{avg}}$
		\end{algorithmic}
	\end{algorithm}

	
	%
	%
	
	\vspace{-0.75em}
	
	\section{Challenges of Designing Robust Aggregation Rules in Decentralized Networks}
	\label{section:challenge-of-robust-aggregations}
	
	Although many existing base aggregators have been shown effective in Byzantine-resilient distributed stochastic optimization, directly extending them to decentralized faces two new challenges:
	(1) issue of disagreement; (2) issue caused by a non-doubly stochastic virtual mixing matrix. Before elaborating on these, we introduce some necessary concepts.
	
	
	
	\vspace{-0.75em}
	
	\subsection{Virtual Mixing Matrix and Contraction Constant}
	\label{section:robust-aggregation-and-mixing-matrix}
	
	With the robust aggregation rule $\A_n$ shown in \eqref{definition:aggregation-centralized-to-decentralized}, we hope that the Byzantine-resilient decentralized SGD would perform similarly to decentralized SGD without Byzantine workers. That is to say, for any honest worker $n \in \N$, we expect that the output of $\A_n (\vx_n, \{\tilde{\vx}_{m,n}\}_{m\in\N_n\cup\B_n})$ approximates a \textit{virtual} weighted average of messages from its honest neighbors and itself, denoted as
	\begin{equation}
		{\rm virtual~weigthed~average}~~   \bar\vx_n := \sum_{m\in\N_n\cup\{n\}}w_{nm}\vx_m,
	\end{equation}
	where $w_{nm} \geq 0$ represents the \textit{virtual} weight\footnote{Note that this is different to $w^{\prime}_{nm}$ that is \textit{explicitly} assigned to any worker $m$, either honest or Byzantine, in \eqref{rule:DECENTRALIZED SGD}.} that honest worker $n$ \textit{virtually} assigns to its honest neighbor (or itself) $m$. We call a row stochastic matrix\footnote{$W$ is \textit{row stochastic} if: (a) all entries $w_{nm} \in [0,1]$; (b) $\sum_{m=0}^N w_{nm}=1$ for all $n$. It is \textit{doubly stochastic} if both $W$ and $W^\top$ are row stochastic.} $W \in \mathbb{R}^{N \times N}$ a virtual mixing matrix if its $(n,m)$-th entry $w_{nm} \in [0, 1]$ when $m \in \N_n \cup\{n\}$ and $w_{nm} = 0$, otherwise.
	Thus, we  can associate the set of robust aggregation rules $\{\A_n\}_{n\in\N}$ with a virtual mixing matrix that is essential to the later algorithm design and analysis.

	However, since the number and identities of the Byzantine workers are unknown, the outputs of $\{\A_n\}_{n\in\N}$ are often biased from the weighted averages $\bar\vx_n$. We further introduce a contraction constant to characterize the biases. The formal definitions are given as follows\footnote{For notational convenience, here we assume that the honest workers are numbered from $1$ to $N$. This can be easily extended to general cases.}.
	
	
			\begin{Definition}[Virtual mixing matrix and contraction constant associated with $\{\A_n\}_{n\in\N}$]
				\label{definition:mixing-matrix}
				Consider a matrix $W \in \mathbb{R}^{N \times N}$ whose $(n,m)$-th entry $w_{nm} \in [0, 1]$ if $m \in \N_n \cup\{n\}$ and $w_{nm} = 0$, otherwise.
				Further, $\sum_{m\in\N_n\cup\{n\}}w_{nm}=1$ for any $n \in \N$. Define $\bar\vx_n := \sum_{m\in\N_n\cup\{n\}}w_{nm}\vx_m$. If there exists a constant $\rho \geq 0$ for any $n \in \N$ such that
				\begin{align}
					\label{inequality:robustness-of-aggregation-local}
					& \|\A_n (\vx_n, \{\tilde{\vx}_{m,n}\}_{m\in\N_n\cup\B_n} )-\bar\vx_n \| \\
					\le & \rho \max_{m\in \N_n\cup\{n\}}\|\vx_m - \bar\vx_n\|, \nonumber
				\end{align}
				then $W$ is the virtual mixing matrix and $\rho$ is the contraction constant associated with $\{\A_n\}_{n\in\N}$.
			\end{Definition}

	\begin{Remark}
		For a set of robust aggregation rules $\{\A_n\}_{n\in\N}$, $\rho$ and $W$ may not be unique, both of which affect convergence and asymptotic learning error. We will formally analyze those effects in Section \ref{section:convergence-learning-error}. Given $\{\A_n\}_{n\in\N}$, determining the best pair of $\rho$ and $W$ is, however, beyond the scope of this paper. We leave it for the future work.
	\end{Remark}

	\vspace{-2em}
	
	\subsection{Disagreement}
	\label{section:challenge-disagreement}
	
	Next we demonstrate that several base aggregators developed for the distributed scenario in  Section \ref{sec:rar} may cause disagreement when extended to the decentralized scenario. That is to say, honest workers are never able to reach the same model in the worst case. Below we give an example.
	
		%
	
	\label{section:shortcoming-majority}
	\begin{figure}
		\centering
		\includegraphics[width=0.32\linewidth]{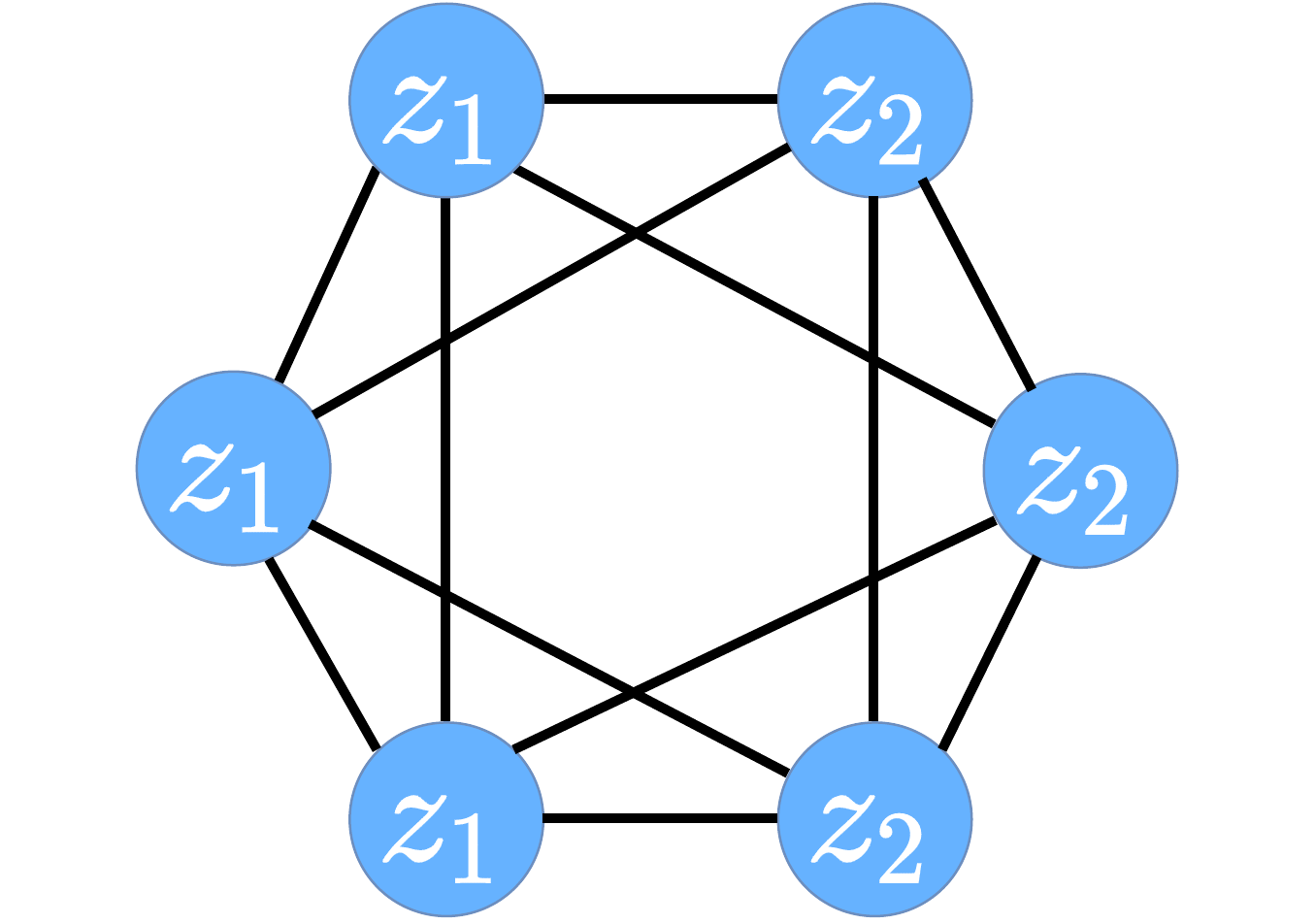}
		\caption{Two-castle graph of $6$ honest workers. Workers marked with $z_1$ are numbered as $1$, $2$ and $3$; marked with $z_2$ are numbered as $4$, $5$ and $6$.}
		\label{fig:TwoCastle}
		\vspace{-0.3cm}
	\end{figure}
	
	\textbf{Two-castle problem.} Consider an undirected graph consisting of $6$ honest workers and no Byzantine workers, as shown in Fig. \ref{fig:TwoCastle}. For simplicity, let $D=1$ such that all the local models $\vx_n$ are scalars. The local cost functions and the local models at time $k$ are respectively given by
	\begin{align}
		f_n(\vx) = \begin{cases}
			(\vx-z_1)^2, ~n =1, 2, 3,\\
			(\vx-z_2)^2, ~n =4, 5, 6,
		\end{cases}\!\!
		\vx^k_n = \begin{cases}
			z_1, ~n = 1, 2, 3,\\
			z_2, ~n = 4, 5, 6,
		\end{cases} \nonumber
	\end{align}
	where $z_1 \neq z_2$ are two constants. We consider the deterministic case so that there is no random variable $\xi_n$ in the argument of $f_n$. Each worker $n$ applies the robust aggregation rule $\A_n$ in \eqref{definition:aggregation-centralized-to-decentralized}. For illustration purpose, we set the base aggregator $\A$ as CooMed to investigate the behavior of such a graph.
	
	
	Since $\nabla f_n(\vx_n^k) = 0$ for all workers $n$, it holds that $\vx_n^{k+\frac{1}{2}} = \vx_n^{k} - \alpha^{k} \nabla f_n(\vx^{k}_n)=\vx_n^k$ according to \eqref{rule:gradient-descent}. Then according to \eqref{rule:SGD-robust} and \eqref{definition:aggregation-centralized-to-decentralized}, given any $r_n \in [0,1)$, it is straightforward to see that for worker $1$, we have
	\begin{align}
		\hspace{-0.5em}     & \vx^{k+1}_1 \\
		\hspace{-0.5em}    =& (1-r_n) \mathrm{CooMed}(\vx^{k+\frac{1}{2}}_1, \vx^{k+\frac{1}{2}}_2, \vx^{k+\frac{1}{2}}_3, \vx^{k+\frac{1}{2}}_5, \vx^{k+\frac{1}{2}}_6) + r_n \vx_n^{k+\frac{1}{2}} \nonumber\\
		\hspace{-0.5em}    =& (1-r_n) \mathrm{CooMed}(z_1, z_1, z_1, z_2, z_2)+ r_n z_1 = z_1 = \vx^k_1, \nonumber
	\end{align}
	which means that no update happens on worker $1$. The same phenomenon can be observed on other workers, as
	\begin{align}
		\vx^{k+1}_n=& (1-r_n) \mathrm{CooMed}(z_1, z_1, z_1, z_2, z_2)+ r_n z_1
		\\
		=& \vx^k_n, ~~n = 1, 2, 3, \nonumber \\
		\vx^{k+1}_n=& (1-r_n) \mathrm{CooMed}(z_1, z_1, z_2, z_2, z_2)+ r_n z_2
		\\
		=& \vx^k_n, ~~n = 4, 5, 6, \nonumber
	\end{align}
	implying that the workers cannot reach consensus forever.
	
	
	
	In fact, a wide class of base aggregators used in Byzantine-resilient distributed stochastic optimization, such as GeoMed and Krum, also suffer from the same disagreement issue.
	Therefore, directly extending the existing base aggregators from Byzantine-resilient distributed stochastic optimization to the decentralized scenario may not work.

	\vspace{-1em}
	
	\subsection{Non-doubly Stochastic Virtual Mixing Matrix}
	
	
	Observe that Definition \ref{definition:mixing-matrix} in Section \ref{section:robust-aggregation-and-mixing-matrix} only requires the virtual mixing matrix $W$ to be row stochastic, rather than doubly stochastic. For example, an equal-weight virtual mixing matrix with $w_{nm}=\frac{1}{N_n+1}$ corresponds to many robust aggregation rules, but for an incomplete network it is not doubly stochastic in general. We give an example to show the asymptotic learning error brought by a non-doubly stochastic virtual mixing matrix.

		%
	

	
	Consider a set of `ideal' robust aggregation rules $\{\A_n\}_{n\in\N}$ with contraction constant $\rho = 0$ in \eqref{inequality:robustness-of-aggregation-local}. That is to say, for honest worker $n \in \mathcal{N}$, the output of $\A_n$ is exactly the weighted average of the messages from its honest neighbors and itself. The associated virtual mixing matrix $W$, however, is not necessarily doubly stochastic. We again illustrate its negative effect via a   deterministic example. Consider the local function of honest worker $n \in \mathcal{N}$ as
	\begin{align}
		f_n(\vx) = \frac{1}{2}\|\vx-\vz_n\|^2,
	\end{align}
	where $\vz_n \in \mathbb{R}^D$ is a constant vector. Therefore, minimizing $\frac{1}{N} \sum_{n\in\N} f_n(\vx)$ gives $\vx^*=\frac{1}{N}\sum_{n\in\N}\vz_n$.
	
	For each honest worker $n \in \mathcal{N}$, the update \eqref{rule:SGD-robust} is given by
	\begin{align}
		\label{rule:SGD-robust-asymmetric-example}
		\vx_n^{k+1} =& \sum_{m\in\N_n\cup\B_n\cup\{n\}} w_{nm} (\vx^k_m- \alpha^k (\vx^k_m-\vz_m)).
	\end{align}
	To write it in a compact form, define 
	\begin{align}
		\label{definition:stacked-X}
		X^k :=& [\vx^k_1, \cdots, \vx^k_n, \cdots, \vx^k_N]^\top \in \mathbb{R}^{N \times D}, \\
		Z :=& [\vz_1, \cdots, \vz_n, \cdots, \vz_N]^\top \in \mathbb{R}^{N \times D}.
	\end{align}
	With these definitions, \eqref{rule:SGD-robust-asymmetric-example} becomes
	\begin{align}
		\label{rule:SGD-robust-asymmetric-example-matrix}
		X^{k+1} = (1-\alpha^k) W X^k + \alpha^k W Z.
	\end{align}
	
	Consider a left eigenvector $\vp$ of $W$ corresponding to eigenvalue $1$ (that is, $\vp^\top W=\vp^\top$). Since $\vp$ is nonnegative and nonzero \cite[Theorem 8.3.1]{horn2012matrix}, we normalize it so that $\1^\top\vp=1$ where $\1$ is a $D$-dimensional all-one vector. We multiply both sides of \eqref{rule:SGD-robust-asymmetric-example-matrix} by $\vp^\top$ and denote the weighted average as $\vy^k := \sum_{n\in\N}p_n\vx_n^k$, then it holds that
	\begin{align}
		\label{rule:SGD-robust-asymmetric-example-weighted}
		\vy^{k+1}= (1-\alpha^k) \vy^k +\alpha^k \sum_{n\in\N}p_n\vz_n.
	\end{align}
	If the step size $\alpha^k$ is properly chosen, there is a unique fixed point of \eqref{rule:SGD-robust-asymmetric-example-weighted}, given by $\vy^\infty := \sum_{n\in\N}p_n\vz_n$.
	In other words, Algorithm \ref{algorithm:ByrdDec} actually optimizes a convex combination of the honest local cost functions. Only when $\vp=\frac{1}{N}\1$, Algorithm \ref{algorithm:ByrdDec} with a proper diminishing step size $\alpha^k$ can optimize the original problem \eqref{problem}. For this case, substituting $\vp=\frac{1}{N}\1$ into $\vp^\top W=\vp^\top$, we can observe that $W^\top$ is also row stochastic, and consequently, $W$ is doubly stochastic.
	
	
	The fact revealed by the above example is not surprising. In Byzantine-free decentralized (stochastic) optimization, we often prefer doubly stochastic mixing matrices. Otherwise, we only minimize the convex combination of local cost functions, instead of the average \cite{lin2021finite}. Existing methods such as push-sum \cite{nedic2014distributed} and push-pull \cite{pu2020push} cope with non-doubly stochastic mixing matrices with correction techniques, but they incur additional transmissions that increase the difficulty of defending against Byzantine attacks. For Byzantine-resilient decentralized deterministic optimization, \cite{su2015fault} shows that a trimmed mean-based algorithm converges to an area determined by the convex combination of honest local cost functions. Our example is more general, covering various Byzantine-resilient decentralized algorithms that can be characterized by virtual mixing matrices $W$ and contraction constants $\rho$.
	
	
	In the next section, we will formally show the influence of disagreement and non-doubly stochastic virtual mixing matrix on the asymptotic learning error.

	\vspace{-0.3em}
	\section{Convergence and Asymptotic Learning Error}
	\vspace{-0.1em}
	\label{section:convergence-learning-error}
	
	In this section, we establish the convergence of the generic Byzantine-resilient decentralized SGD in Algorithm \ref{algorithm:ByrdDec} and identify the factors that determine the asymptotic learning error. We begin with several assumptions.
	\vspace{-0.1em}
	
	
	\begin{Assumption}[Lower boundedness]
		\label{assumption:lower-bound}
		The  aggregated cost function $f(\vx)$ is lower bounded by $f^*$; i.e., $f(\vx)\ge f^*, \forall\vx$.
	\end{Assumption}
	\begin{Assumption}[$L$-smoothness]  \label{assumption:Lip}
		For each honest worker $n\in\N$, the local cost function $f_n(\vx;\xi_n)$ is $L$-smooth.
	\end{Assumption}

	\begin{Assumption}[Bounded inner variation]
		\label{assumption:innerVariance}
		For any honest worker $n\in\N$ and $\vx$, the variation of its stochastic gradients with respect to its aggregated gradient is bounded by
		\begin{align}
			\label{eq:innerVariance}
			\E_{\xi_n}[\|\nabla f_n(\vx; \xi_n)-\nabla f_n(\vx)\|^2]\le\delta^2_{\rm in}.
		\end{align}
	\end{Assumption}
	
	\begin{Assumption}[Bounded outer variation]
		\label{assumption:outerVariance}
		For any $\vx \in \mathbb{R}^D$, the variation of the aggregated gradients
		at the honest workers with respect to the overall aggregated gradient
		is upper-bounded by
		\begin{align}
			\label{eq:outerVariance}
			\max_{n\in\N} \|\nabla f_n(\vx)-\nabla f(\vx)\|^2 \le\delta^2_{\rm out}.
		\end{align}
	\end{Assumption}
	
	\begin{Assumption}[Independent sampling]
		\label{assumption:indSampling}
		The stochastic gradients $\nabla f_n(\vx^{k}_n; \xi_n^{k})$ are independently sampled over times $k=$ $0, 1, \ldots$ and across honest workers $n \in \N$.
	\end{Assumption}
	
	Assumptions \ref{assumption:lower-bound} and \ref{assumption:Lip} are common in (stochastic) gradient-based non-convex optimization.
	Assumptions \ref{assumption:innerVariance} and \ref{assumption:outerVariance} bound the variation of stochastic gradients on each honest worker and the variation of aggregated gradients across the honest workers, respectively. Assumption \ref{assumption:indSampling} guarantees that the stochastic gradients are independent. They are standard in the analysis of stochastic optimization.
	
	\vspace{-0.1em}
	\begin{Theorem}[Convergence]
		\label{theorem:convergence}
		Consider the Byzantine-resilient decentralized SGD in Algorithm \ref{algorithm:ByrdDec}. Suppose that the robust aggregation rules $\{\A_n\}_{n\in\N}$ satisfy \eqref{inequality:robustness-of-aggregation-local}. With $\bar\vx^k := \frac{1}{N}\sum_{n\in\N}\vx^k_n$ denoting the average\footnote{\blue{Note that this is different to $\bar\vx_n^k := \sum_{m\in\N_n\cup\{n\}}w_{nm}\vx_m^k$, the weighted average of honest neighboring and own models for worker $n$ in Definition \ref{definition:mixing-matrix}.}} of all honest models at time $k$,
		define the disagreement measure $H^k$ as
		\blue{
			\begin{align}
				\label{definition:Hk}
				H^k := \frac{1}{N}\sum_{n\in\N} \|\vx^{k}_{n}-\bar\vx^{k}\|^2.
			\end{align}
		}\!\!
		\vspace{0.5em}
		Under Assumptions \ref{assumption:lower-bound}--\ref{assumption:indSampling}, if a constant step size $\alpha^k=\alpha \le \frac{1}{2\sqrt{3}L}$ is used, then it holds that
		\begin{align}
			\label{inequality:convergence-0}
			\frac{1}{K}\sum_{k=1}^{K}\E[\|&\nabla f(\bar\vx^k)\|^2]
			\le \frac{2(f(\bar\vx^0)-f^*)}{\alpha K}
			+ \frac{2\alpha\delta^2_{\rm in}L}{N}
			\\
			&
			+\frac{36(\rho^2N + \chi^2)+3\alpha^2L^2}{\alpha^2K}
			\sum_{k=1}^{K}\E [H^k] \nonumber\\
			&+96({\rho^2 N} +\chi^2 )(\delta^2_{\rm in}+\delta^2_{\rm out}), \nonumber
		\end{align}
		where the expectation is taken over all random variables $\vxi^{0}$, $\vxi^{1}$, $\cdots$, $\vxi^{K}$, and $\chi^2 := \frac{1}{N}\|W^\top\1-\1\|^2$ describes how non-doubly stochastic $W$ is.
		Additionally, if the step size $\alpha^k$ is set as $\alpha= {O}(\sqrt{{N}/{(\delta^2_{\rm in}K)}})$, then it holds that
		\begin{align}
			\label{inequality:convergence}
			\frac{1}{K}&\sum_{k=1}^{K}\E[\|\nabla f(\bar\vx^k)\|^2]
			\le O\lp\sqrt{\frac{\delta^2_{\rm in}}{NK}}\rp
			+ O \lp \frac{L^2}{K} \rp \sum_{k=1}^{K}\E [H^k] \nonumber
			\\
			&
			+{O}(
			\rho^2 N + \chi^2)
			\lp \frac{\delta^2_{\rm in}}{N}\sum_{k=1}^{K}\E [H^k]+\delta^2_{\rm in} +\delta^2_{\rm out}\rp.
		\end{align}
	\end{Theorem}
	The proof is deferred to Section \ref{section:proof-theorem-convergence}. Theorem \ref{theorem:convergence} asserts that the time-averaged squared gradient norm, computed on the averages of all honest models $\bar\vx^k := \frac{1}{N}\sum_{n\in\N}\vx^k_n$, is upper-bounded by the summation of the three terms at the right-hand side (RHS) of \eqref{inequality:convergence}. Among them, the first term vanishes at the rate of $O (\frac{1}{\sqrt{K}})$ which also appears in the convergence analysis of distributed/decentralized SGD.
	Observe that when the accumulated expected disagreement measure $\sum_{k=1}^{K}\E [H^k]$ in the second term is unbounded, the third term is unbounded too, and the algorithm is not Byzantine-resilient.
	
	Therefore, at this stage, we make a hypothesis that the accumulated expected disagreement measure is upper-bounded by a constant asymptotically. In Theorem \ref{theorem:consensus}, we will discuss when this hypothesis holds true.
	With $K\to\infty$, \eqref{inequality:convergence} gives the \textit{asymptotic learning error} of the Byzantine-resilient decentralized SGD in Algorithm \ref{algorithm:ByrdDec}, defined as
	\begin{align}
		\label{inequality:asymtotic-learning-error}
		&\limsup_{K\to\infty} \frac{1}{K}\sum_{k=1}^{K}\E[\|\nabla f(\bar\vx^k)\|^2]
		\\
		\le &
		{O}(
		\underbrace{\rho^2 N}_{\text{Est. Err.}}
		+ \underbrace{\chi^2}_{\text{Mix. Err.}}
		)
		\big( \frac{\delta^2_{\rm in}}{N} \underbrace{  \limsup_{K\to\infty} \sum_{k=1}^{K}\E [H^k] }_{\text{Con. Err.}}+
		\delta^2_{\rm in} +\delta^2_{\rm out}\big),
		\nonumber
	\end{align}
	where the asymptotic learning error is determined by three factors associated with robust aggregation rules: estimation, mixing, and consensus errors.
	
	Essentially, the asymptotic learning error arises from the inaccurate estimation of the overall aggregated gradient.
	This is also the case in the decentralized scenario.  Below we discuss three errors contributing to the asymptotic learning error.
	
	\noindent \textbf{Estimation error.} The error related to $\rho^2 N$ reflects the bias caused by Byzantine attacks and robust aggregation rules. When the Byzantine workers are absent and a proper aggregation rule is used, this error turns to $0$.
	
	\noindent \textbf{Mixing error.} The error related to $\chi^2$ comes from the non-doubly stochastic virtual mixing matrix $W$. This error becomes $0$ when $W$ is doubly stochastic.


	\noindent \textbf{Consensus error.} The effect of the consensus error, or formally the asymptotic accumulated expected disagreement, on the asymptotic learning error is added by the stochastic gradient noise (inner variation $\delta^2_{\rm in}$), the data heterogeneity (outer variation $\delta^2_{\rm out}$), and is further amplified by $\rho^2 N + \chi^2$.
	
	The estimation error also appears in the distributed scenario, whereas the mixing and consensus errors are unique in the decentralized scenario.
	
	
	\begin{Remark}
		The effect of the inner variation $\delta^2_{\rm in}$ on the asymptotic learning error could be removed by variance reduction \cite{wu2020federated, khanduri2019byzantine} and momentum \cite{karimireddy2021learning} techniques, via replacing the gradient descent step in line 4 of Algorithm \ref{algorithm:ByrdDec} with a corrected gradient descent update.
	\end{Remark}
	
	Recall that the asymptotic learning error in \eqref{inequality:asymtotic-learning-error} is meaningful only when the accumulated expected disagreement measure $\sum_{k=1}^{K}\E [H^k]$ is upper-bounded by a constant asymptotically. However, this hypothesis may not hold in general -- we have shown in Section \ref{section:challenge-disagreement} that $H^k$ can be a constant for many existing robust aggregation rules even without Byzantine workers. Below, we show when this hypothesis holds true.
	
	
	Apparently, if some honest workers cannot communicate with others, reaching consensus is impossible. Therefore, it is natural to make an assumption on the network connectivity. Consider a (possibly directed) graph $\mathcal{G}_W := (\N, \edge_W)$ that corresponds to the virtual mixing matrix $W$, with $(n,m)\in\edge_W$ if and only if $w_{nm}>0$. Thus, $\mathcal{G}_W$ is a subgraph of $\mathcal{G}$.
	Now we introduce an assumption about $\mathcal{G}_W$ and $W$.
	
	
	\begin{Assumption}[Network connectivity]
		\label{assumption:connection}
		The graph $\mathcal{G}_W$ is strongly connected, which means that for any pair $n, m\in\N$, there exists at least one directed path between them.
		In addition, $\lambda := 1-\|(I-\frac{1}{N}\1\1^\top)W\|^2>0$,
		where $\|\cdot\|$ is the matrix spectral norm.
	\end{Assumption}
	
	

	The following theorem establishes the upper bound of the expected disagreement measure $\E [H^k]$ when the contraction constant $\rho$ is sufficiently small.

	\begin{Theorem}[Consensus of robust aggregation rules]
		\label{theorem:consensus}
		Consider the Byzantine-resilient decentralized SGD in Algorithm \ref{algorithm:ByrdDec}. Suppose that the robust aggregation rules $\{\A_n\}_{n\in\N}$ satisfy \eqref{inequality:robustness-of-aggregation-local} and the contraction factor $\rho$ satisfies
		\begin{align}
			\label{condition:rho}
			\blue{\rho < \rho^* := \frac{\lambda}{8\sqrt{N}}.}
		\end{align}
		Define a constant $\omega := \lambda-8\rho \sqrt{N}$.
		Under Assumptions \ref{assumption:Lip}--\ref{assumption:connection},
		if a constant step size $\alpha^k = \alpha \le \frac{1}{3L} \sqrt{(2-\omega) \omega^2 /(6-2\omega)}$ is used
		then it holds that
		\begin{align}
			\label{inequality:H-convergence}
			\E [H^k]
			\le & \alpha^2 \Delta (\delta^2_{\rm in} +\delta^2_{\rm out})
		\end{align}
		where $\E$ is taken over all $\vxi^{0}, \vxi^{1}, \cdots, \vxi^{k}$ and the constant $\Delta := \frac{12(1-\omega)}{\omega^3}$.
		Additionally, if the step size $\alpha^k$ is set as $\alpha= {O}(\sqrt{{N}/{(\delta^2_{\rm in}K)}})$,
		the consensus error is bounded by
		\begin{align}
			\label{add-0001}
			\limsup_{K\to\infty} \sum_{k=1}^{K}\E [H^k] \le
			O\lp  \frac{\Delta N (\delta^2_{\rm in} +\delta^2_{\rm out})}{\delta^2_{\rm in}}\rp.
		\end{align}
	\end{Theorem}

	The proof is deferred to Appendix \ref{section:proof-theorem-consensus}.
	According to Theorem \ref{theorem:consensus}, if we choose the same step size
	as in Theorem \ref{theorem:convergence}, the expected disagreement measure is in the order of ${O}(\frac{1}{K})$, such that the consensus error is upper-bounded. Consequently, the asymptotic learning error in \eqref{inequality:asymtotic-learning-error} is upper-bounded by
	\begin{align}
		\label{inequality:asymtotic-learning-error-2}
		& \limsup_{K\to\infty} \frac{1}{K}\sum_{k=1}^{K}\E[\|\nabla f(\bar\vx^k)\|^2]  \\
		\leq &
		{O}(
		\rho^2 N + \chi^2
		)\lp 1+\Delta \rp
		(\delta^2_{\rm in} +\delta^2_{\rm out}). \nonumber
	\end{align}

	\begin{Remark}
		The bound in \eqref{inequality:H-convergence} recovers the traditional results of Byzantine-free decentralized deterministic optimization \cite{yuan2016convergence} and stochastic optimization \cite{koloskova2020unified} when $\rho = 0$. In this case, the asymptotic disagreement  is ${O}(\alpha^2(\delta^2_{\rm in} +\delta^2_{\rm out}))$, which depends on the step size $\alpha$ and the sum of inner and outer variations $\delta^2_{\rm in} +\delta^2_{\rm out}$.
		Theorem \ref{theorem:consensus} also shows the difficulty of reaching consensus under Byzantine attacks, even when the network topology and the virtual mixing matrix are both perfect. In a fully connected network   and with an equal-weight virtual mixing matrix $W$, we have $\lambda=1$, $\omega=1-8\rho \sqrt{N}$,
		and thus
		\begin{align}
			\E [H^k] \le O\lp\frac{\rho}{(1-8\rho \sqrt{N})^3}\rp,
		\end{align}
		implying that the asymptotic disagreement measure could be large when $\rho$ is close to $\frac{1}{8 \sqrt{N}}$.
	\end{Remark}
	
	\section{Guidelines of Designing Robust Aggregation Rules for a Decentralized Network}
	\label{section:algorithm_design}
	According to the analysis in Section V, we give the guidelines of designing robust aggregation rules that are suitable for a decentralized network.
	Further, we show that the design of IOS exactly follows the guidelines.
	
	\subsection{Design Guidelines for Robust Aggregation Rules}
	\label{section:RCA}
	
	Theorem \ref{theorem:convergence} shows that to reduce the estimation and mixing errors, the robust aggregation rules should have a small contraction constant $\rho$ and a doubly stochastic virtual mixing matrix $W$. With a small $\rho$, we can also bound the consensus error by \eqref{add-0001} in Theorem \ref{theorem:consensus}.
	Thus, we design robust aggregation rules to satisfy the following conditions.

	\begin{Definition}[Robust contractive aggregation (RCA)]
		\label{definition:robust-contraction-aggregation}
		A set of aggregation rules $\{\A_n\}_{n\in\N}$ are RCA if the associated contraction constant $\rho$ satisfies \eqref{condition:rho}.
	\end{Definition}
	
	\begin{Definition}[Robust doubly stochastic aggregation (RDSA)]
		\label{definition:robust-ds-aggregation}
		A set of aggregation rules $\{\A_n\}_{n\in\N}$ are RDSA if the associated virtual mixing matrix $W$ is doubly stochastic.
	\end{Definition}

	
	To better understand Definition \ref{definition:robust-contraction-aggregation}, we compute the contraction constant $\rho$ in a fully connected network for several existing robust aggregation rules extended to the decentralized scenario; see Table \ref{table:rho}.
	Therein, $\mu := \frac{B}{N+B}$ is the proportion of Byzantine workers.
	When $\mu$ is sufficiently small, FABA, trimmed mean (TriMean) and CC/SCC have sufficiently small contraction constants $\rho$ that may satisfy \eqref{condition:rho}. However, those of CooMed, GeoMed and Krum are always greater than $1$. Therefore, CooMed, GeoMed and Krum are not RCA. This is consistent with the fact that they have the disagreement issue in Section \ref{section:challenge-disagreement}. For Definition \ref{definition:robust-ds-aggregation}, we consider a general incomplete network and check whether the above robust aggregation rules are RDSA; see Table \ref{table:rho}.
	
	From Table \ref{table:rho}, we observe that most existing decentralized robust aggregation rules and those extended from the distributed scenario to decentralized are not necessarily RCA or RDSA. The only exception is SCC \cite{he2022byzantine}. However, SCC needs to set task-dependent and time-varying thresholds, whose choices rely on the estimates of the true models; see the discussion in Appendix \ref{section:coverage-of-generic-form}.
	In practice, these thresholds are often set as constants, which explains why SCC may not perform well as is indicated by the theory. In contrast, for IOS, the parameters $\{q_n\}_{n \in \mathcal{N}}$ stand for the estimates of the numbers of Byzantine neighbors, which are relatively easy to obtain. In the numerical experiments, we will show that IOS outperforms SCC under most Byzantine attacks.
	
	Therefore, we provide the following guidelines for designing decentralized stochastic algorithms.
	
	\begin{tcolorbox}[width=0.49\textwidth]
		\vspace{-0.5em}
		{\bf Guidelines:} For a decentralized stochastic optimization algorithm to be Byzantine-resilient, the set of aggregation rules $\{\A_n\}_{n\in\N}$ should be RCA and RDSA.
		\vspace{-1.5em}
	\end{tcolorbox}
	
	\blue{We summarize the results when a set of robust aggregation rules $\{\A_n\}_{n\in\N}$ are both RCA and RDSA in the next corollary.
		
		\begin{Corollary}
			\label{corollary:new}
			Consider the Byzantine-resilient decentralized SGD in Algorithm \ref{algorithm:ByrdDec}. Suppose that the robust aggregation rules $\{\A_n\}_{n\in\N}$ satisfy \eqref{inequality:robustness-of-aggregation-local}, and are both RCA and RDSA.
			Under Assumptions \ref{assumption:lower-bound}--\ref{assumption:connection}, if the step size $\alpha^k$ is set as $\alpha= {O}(\sqrt{{N}/{(\delta^2_{\rm in}K)}})$,
			then the consensus error is  bounded by
			\begin{align}
				\label{add-0001-new}
				\limsup_{K\to\infty} \sum_{k=1}^{K}\E [H^k] \le
				O\lp  \frac{\Delta N (\delta^2_{\rm in} +\delta^2_{\rm out})}{\delta^2_{\rm in}}\rp,
			\end{align}
			and the asymptotic learning error is upper-bounded by
			\begin{align}
				\label{inequality:asymtotic-learning-error-2-new}
				&\limsup_{K\to\infty} \frac{1}{K}\sum_{k=1}^{K}\E[\|\nabla f(\bar\vx^k)\|^2] \\
				\leq &
				{O}(
				\rho^2 N
				)\lp 1+\Delta \rp
				(\delta^2_{\rm in} +\delta^2_{\rm out}). \nonumber
			\end{align}
	\end{Corollary}}
	
	When the robust aggregation rules $\{\A_n\}_{n\in\N}$ are RCA, $\rho < \frac{\lambda}{8 \sqrt{N}}$. In consequence, the asymptotic learning error in \eqref{inequality:asymtotic-learning-error-2-new} is in the order of ${O}(\delta^2_{\rm in} +\delta^2_{\rm out})$, which matches the bound of the distributed scenario \cite{wu2020federated}.

	\begin{table}[tbp]
		\small
		\caption{Robust aggregation rules, their corresponding $\rho$ and whether they are RCA in a complete network, as well as whether they are RDSA in a general incomplete network}
		\label{table:rho}
		\centering
		\begin{tabular}{cccc}
			\hline
			Aggregation & $\rho$ & RCA & RDSA \\
			\hline
			CooMed \cite{yin2018ByzantineRobustDL} & ${O}(1+\mu)$ & $\times$ & $\times$ \\
			GeoMed \cite{chen2019DistributedSM, xie2018GeneralizedBS} & ${O}(\frac{1-\mu}{1-2\mu}$)  & $\times$ & $\times$ \\
			Krum \cite{blanchard2017MachineLW} & ${O}(N+B-q)$ & $\times$ & $\times$ \\
			CC \cite{karimireddy2021learning} & ${O}(\mu)$ & $\surd$ & $\times$ \\
			FABA \cite{xia2019faba} & ${O}(\frac{\mu}{1-3\mu})$ & $\surd$ & $\times$ \\
			\hline
			TriMean \cite{xie2018phocas,su2015fault,su2020byzantine,fang2022bridge} & ${O}(\frac{\mu(1-\mu)}{(1-2\mu)^2})$ & $\surd$ & $\times$ \\
			SCC \cite{he2022byzantine} & ${O}(\mu)$ & $\surd$ & $\surd$ \\
			\textbf{IOS (ours)} & ${O}(\frac{\mu}{1-3 \mu})$ & $\surd$ & $\surd$\\
			\hline
		\end{tabular}
	\end{table}

	\subsection{IOS is both RCA and RDSA}
	
	Before showing that the robust aggregation rules generated by IOS are simultaneously RCA and RDSA, we observe that each honest worker $n \in \N$ discards $q_n$ models from the received ones, along with the  weights $w_{nm}'$. We define a set that includes the neighbors with the largest $q_n$ weights, as
	\begin{align}
		\cU^{\rm max}_n := \underset{\cU: \cU \subseteq \N_n\cup\B_n \atop |\cU|=q_n}{\arg\max} \sum_{m \in \cU}w_{nm}'.
	\end{align}
	According to the notation in \eqref{eq:wprime}, ${\ccalW'_n}({\cU^{\rm max}_n})$ accumulates the largest $q_n$ weights that worker $n$ assigns to its neighbors.
	
	\blue{The following theorem shows the contraction factor and the virtual mixing matrix associated with IOS. The analysis holds for both complete and general incomplete networks.}
	%
	%
	
	\begin{Theorem}[Contraction factor and virtual mixing matrix of IOS]
		\label{theorem:robustness-of-aggregation}
		If $q_n$ is chosen such that $B_n\le q_n$ and ${\ccalW'_n}({\cU^{\rm max}_n})<\frac{1}{3}$, then for the robust aggregation rules generated by IOS in Algorithm \ref{algorithm:ios}, the associated virtual mixing matrix $W$ is doubly stochastic and its $(n,m)$-th entry is given by
		\begin{align}
			w_{nm} = \begin{cases}
				w_{nn}'+\sum_{b\in\B_n} w_{nb}', &m=n,\\
				w_{nm}', &m\neq n,\\
			\end{cases}
		\end{align}
		while the contraction constant is bounded by
		\begin{align}
			\label{equality:rho-of-ios}
			\rho \le \max_{n\in\N}\frac{12 {\ccalW'_n}( {\cU^{\rm max}_n})}{1-3 {\ccalW'_n}( {\cU^{\rm max}_n})}.
		\end{align}
		Therefore, the robust aggregation rules generated by IOS in Algorithm \ref{algorithm:ios} are RDSA, and further are RCA when
		\begin{align}
			\label{condition:h}
			{\ccalW'_n}({\cU^{\rm max}_n}) < \frac{\rho^*}{12+3\rho^*}, \quad \forall n \in \N.
		\end{align}
	\end{Theorem}
	
	\begin{table*}[!tbp]
		\small
		\caption{Accuracy (Acc) and disagreement measure (DM) in the two-castle graph for the i.i.d. case.}
		\label{table:two-castle-iid}
		\centering
		\setlength{\tabcolsep}{1.0mm}{
			\begin{tabular}{c|cc|cc|cc|cc|cc|cc}
\hline\hline
\multirow{2}{*}{}&\multicolumn{2}{c|}{no attack}&\multicolumn{2}{c|}{Gaussian}&\multicolumn{2}{c|}{sign-flipping}&\multicolumn{2}{c|}{isolation}&\multicolumn{2}{c|}{sample-duplicating}&\multicolumn{2}{c}{ALIE}\\
& Acc.(\%) & CE & Acc.(\%) & CE & Acc.(\%) & CE & Acc.(\%) & CE & Acc.(\%) & CE & Acc.(\%) & CE \\
\hline
no comm. & 90.24 & $>$1e-01 & -- & --  & -- & --  & -- & --  & -- & --  & -- & -- \\
\hline
WeiMean & \textbf{91.71} & {$<$1e-07} & {16.09} & {$>$1e-01} & {29.23} & {1e-07} & {90.33} & {$>$1e-01} & \textbf{91.71} & {$<$1e-07} & {91.67} & {$<$1e-07}\\
CooMed & {91.49} & {$<$1e-07} & {91.53} & {$<$1e-07} & {87.04} & {$<$1e-07} & {91.40} & {1e-07} & {91.56} & {$<$1e-07} & {91.38} & {$<$1e-07}\\
GeoMed & {91.67} & {$<$1e-07} & \textbf{91.68} & {$<$1e-07} & {87.05} & {$<$1e-07} & {90.16} & {$>$1e-01} & {91.68} & {$<$1e-07} & {91.58} & {$<$1e-07}\\
Krum & {90.72} & {$<$1e-07} & {90.77} & {1e-07} & {91.13} & {1e-07} & {90.72} & {1e-07} & {91.04} & {6e-07} & {91.43} & {$<$1e-07}\\
TriMean & {91.70} & {$<$1e-07} & {91.61} & {$<$1e-07} & {86.10} & {$<$1e-07} & {91.61} & {$<$1e-07} & {91.65} & {$<$1e-07} & {91.59} & {$<$1e-07}\\
SimRew & {76.42} & {$>$1e-01} & {73.96} & {$>$1e-01} & {73.91} & {$>$1e-01} & {74.00} & {$>$1e-01} & {73.92} & {$>$1e-01} & {73.99} & {$>$1e-01}\\
DRSA & {91.68} & {2e-06} & {91.65} & {2e-06} & {89.95} & {3e-06} & {91.60} & {7e-06} & {91.65} & {2e-06} & {91.65} & {2e-06}\\
CC & {91.67} & {$<$1e-07} & {91.65} & {2e-06} & {29.86} & {1e-07} & {91.56} & {6e-06} & {91.68} & {$<$1e-07} & \textbf{91.70} & {$<$1e-07}\\
SCC & {91.70} & {$<$1e-07} & \textbf{91.68} & {2e-06} & {35.69} & {1e-07} & {91.62} & {6e-06} & {91.66} & {$<$1e-07} & {91.67} & {$<$1e-07}\\
FABA & \textbf{91.71} & {$<$1e-07} & \textbf{91.68} & {$<$1e-07} & \textbf{91.65} & {$<$1e-07} & \textbf{91.71} & {$<$1e-07} & {91.63} & {$<$1e-07} & {91.59} & {$<$1e-07}\\
\textbf{IOS (ours)} & {91.69} & {$<$1e-07} & \textbf{91.68} & {$<$1e-07} & \textbf{91.65} & {$<$1e-07} & {91.67} & {$<$1e-07} & {91.67} & {$<$1e-07} & {91.61} & {$<$1e-07}\\
\hline\hline
\end{tabular}

		}
		
		\vspace{5mm}
		\caption{Accuracy (Acc) and disagreement measure (DM) in the two-castle graph for the non-i.i.d. case.}
		\label{table:two-castle-non-iid}
		\centering
		\small
		\setlength{\tabcolsep}{1.0mm}{
			\begin{tabular}{c|cc|cc|cc|cc|cc|cc}
\hline\hline
\multirow{2}{*}{}&\multicolumn{2}{c|}{no attack}&\multicolumn{2}{c|}{Gaussian}&\multicolumn{2}{c|}{sign-flipping}&\multicolumn{2}{c|}{isolation}&\multicolumn{2}{c|}{sample-duplicating}&\multicolumn{2}{c}{ALIE}\\
& Acc.(\%) & CE & Acc.(\%) & CE & Acc.(\%) & CE & Acc.(\%) & CE & Acc.(\%) & CE & Acc.(\%) & CE \\
\hline
no comm. & 10.00 & $>$1e-01 & -- & --  & -- & --  & -- & --  & -- & --  & -- & -- \\
\hline
WeiMean & \textbf{91.73} & {1e-07} & {15.79} & {$>$1e-01} & {48.99} & {5e-07} & {10.00} & {$>$1e-01} & \textbf{91.66} & {3e-07} & {91.11} & {$<$1e-07}\\
CooMed & {77.60} & {6e-07} & {80.23} & {1e-07} & {74.55} & {1e-02} & {33.39} & {$>$1e-01} & {81.03} & {1e-07} & {82.03} & {$<$1e-07}\\
GeoMed & {89.29} & {$<$1e-07} & {89.66} & {$<$1e-07} & {8.36} & {5e-07} & {10.00} & {$>$1e-01} & {89.90} & {2e-07} & {88.63} & {$<$1e-07}\\
Krum & {17.45} & {3e-06} & {19.74} & {4e-07} & {10.10} & {$<$1e-07} & {19.74} & {4e-07} & {29.17} & {5e-07} & {82.62} & {5e-07}\\
TriMean & {91.72} & {$<$1e-07} & {90.08} & {1e-06} & {72.02} & {8e-07} & {58.53} & {2e-06} & {87.21} & {1e-06} & {77.85} & {1e-06}\\
SimRew & {10.34} & {$>$1e-01} & {10.48} & {$>$1e-01} & {10.48} & {$>$1e-01} & {10.48} & {$>$1e-01} & {10.48} & {$>$1e-01} & {10.48} & {$>$1e-01}\\
DRSA & {81.55} & {4e-03} & {89.24} & {4e-03} & {11.35} & {6e-03} & {78.87} & {2e-03} & {86.06} & {5e-03} & {61.37} & {6e-03}\\
CC & {91.67} & {$<$1e-07} & {91.60} & {2e-05} & {48.99} & {5e-07} & {86.48} & {6e-05} & {91.64} & {3e-07} & \textbf{91.16} & {$<$1e-07}\\
SCC & \textbf{91.73} & {1e-07} & {91.55} & {2e-05} & {48.99} & {5e-07} & {86.12} & {6e-05} & {91.65} & {3e-07} & {91.11} & {$<$1e-07}\\
FABA & {91.72} & {$<$1e-07} & \textbf{91.68} & {$<$1e-07} & \textbf{91.66} & {$<$1e-07} & \textbf{91.64} & {1e-07} & {88.90} & {6e-07} & {82.86} & {6e-07}\\
\textbf{IOS (ours)} & \textbf{91.73} & {1e-07} & \textbf{91.68} & {$<$1e-07} & \textbf{91.66} & {$<$1e-07} & \textbf{91.64} & {1e-07} & {88.84} & {6e-07} & {82.78} & {6e-07}\\
\hline\hline
\end{tabular}

		}
	\end{table*}
	
	\begin{table*}[!ht]
		\small
		\centering
		\caption{Accuracy (Acc) and disagreement measure (DM) in the octopus graph for the non-i.i.d. case.}
		\label{table:octopus-non-iid}
		\setlength{\tabcolsep}{1.0mm}{
			\begin{tabular}{c|cc|cc|cc|cc|cc|cc}
\hline\hline
\multirow{2}{*}{}&\multicolumn{2}{c|}{no attack}&\multicolumn{2}{c|}{Gaussian}&\multicolumn{2}{c|}{sign-flipping}&\multicolumn{2}{c|}{isolation}&\multicolumn{2}{c|}{sample-duplicating}&\multicolumn{2}{c}{ALIE}\\
& Acc.(\%) & CE & Acc.(\%) & CE & Acc.(\%) & CE & Acc.(\%) & CE & Acc.(\%) & CE & Acc.(\%) & CE \\
\hline
CC & {89.26} & {4e-06} & {89.17} & {8e-06} & {32.20} & {3e-05} & {62.48} & {2e-05} & {89.23} & {4e-06} & {89.14} & {4e-06}\\
SCC & \textbf{90.19} & {2e-05} & \textbf{90.16} & {3e-05} & \textbf{32.24} & {1e-04} & \textbf{67.80} & {7e-05} & \textbf{90.24} & {2e-05} & \textbf{90.10} & {3e-05}\\
\hline
FABA & {89.26} & {4e-06} & {89.26} & {4e-06} & {89.26} & {4e-06} & {89.24} & {4e-06} & {89.09} & {4e-06} & {88.87} & {4e-06}\\
\textbf{IOS (ours)} & \textbf{90.19} & {2e-05} & \textbf{90.21} & {2e-05} & \textbf{90.21} & {2e-05} & \textbf{90.20} & {2e-05} & \textbf{89.82} & {2e-05} & \textbf{89.49} & {3e-05}\\
\hline
\hline
\end{tabular}

		}
	\end{table*}
	
	\blue{Theorem \ref{theorem:robustness-of-aggregation} implies that when ${\ccalW'_n}({\cU^{\rm max}_n})$ is sufficiently small, the virtual mixing matrix $W$ associated with IOS is doubly stochastic, while the contraction factor $\rho$ is in the order of ${O}(\max_{n\in\N} {\ccalW'_n}({\cU^{\rm max}_n}))$ and can be sufficiently small as well.}
	%
	To understand the condition ${\ccalW'_n}({\cU^{\rm max}_n})<\frac{1}{3}$, consider a complete network and let $w_{nm}'=\frac{1}{N_n+B_n+1}$. In this circumstance, the condition is equivalent to $\frac{q_n}{N_n+B_n+1}<\frac{1}{3}$, meaning that the proportion of Byzantine workers cannot exceed $\frac{1}{3}$. If we further let $q_n = B_n$, \eqref{equality:rho-of-ios} becomes $\rho={O}(\frac{\mu}{1-3\mu})$, which is the result listed in Table \ref{table:rho}. For a general incomplete network, $w_{nm}'=\min\lb\frac{1}{N_n+B_n+1}, \frac{1}{N_m+B_m+1}\rb$ $\le \frac{1}{N_n+B_n+1}$ if $W'$ is constructed by the Metropolis-Hastings rule \cite{xiao2004fast}. Therefore, if $q_n=B_n$ and $\frac{B_n}{N_n+B_n+1}$, the portion of Byzantine neighbors is smaller than $\frac{1}{3}$, then ${\ccalW'_n}({\cU^{\rm max}_n}) \le \frac{B_n}{N_n+B_n+1}<\frac{1}{3}$.

	Having ${\ccalW'_n}({\cU^{\rm max}_n}) <\frac{1}{3}$ guarantees that the robust aggregation rules generated by IOS are RDSA. But to guarantee that they are also RCA, we further require $\rho \leq \rho^*$ according to \eqref{condition:rho}, meaning that ${\ccalW'_n}({\cU^{\rm max}_n}) < \frac{\rho^*}{12+3\rho^*}$. To meet this requirement, the numbers of Byzantine neighbors must be limited and the initial mixing matrix $W'$ must be properly chosen. \blue{Below we give two examples in which \eqref{condition:h} holds.
		
		\textbf{Example 1.} When there are no Byzantine workers, we can let $q_n = B_n = 0$. Therefore, ${\ccalW'_n}({\cU^{\rm max}_n}) = 0$ and \eqref{condition:h} is satisfied. In this case, $\rho = 0$.
		
		\textbf{Example 2.} For an Erdos-Renyi graph, $\lambda = \Theta(1)$  and $\rho^* \leq {O}(\frac{1}{\sqrt{N}})$; see \cite{ying2021exponential,nedic2018network}. As we have discussed below Theorem \ref{theorem:robustness-of-aggregation}, ${\ccalW'_n}({\cU^{\rm max}_n}) \leq {O}(\frac{B_n}{N_n+B_n+1}) \simeq {O}(\frac{B}{N+B+1})$. Thus, \eqref{condition:h} is satisfied if $B \leq {O}(\sqrt{N})$. In this case, $\rho \leq {O}(\frac{B}{N+B+1})$.
		Generally speaking, if a graph has $\lambda = \Theta(N^{-\varepsilon})$ where $0 \leq \varepsilon<\frac{1}{2}$, \eqref{condition:h} is satisfied if $B \leq {O}(N^{\frac{1}{2}-\varepsilon})$.}


	\section{Numerical Experiments}
	\label{section:experiment}
	In this section, we evaluate the performance of the proposed IOS on the softmax regression task. The dataset is MNIST, which contains 60,000 images of handwritten digits from $0$ to $9$.
	The local data distributions are i.i.d. or non-i.i.d. across the workers. In the i.i.d. case, all images in each class are evenly distributed across all workers. In the non-i.i.d. case, each worker only has images from one class. Unless otherwise stated, we use squared $\ell_2$-norm regularization with coefficient $0.01$, step size $\alpha^{k}=0.9/\sqrt{k+1}$ and batch size $32$.
	
	We are going to compare decentralized SGD equipped with various aggregation rules: \emph{weighted mean (WeiMean)} that is not Byzantine-resilient, \emph{coordinate-wise median (CooMed)}, \emph{geometric median (GeoMed)},  \emph{Krum}, \emph{trimmed mean (TriMean)} \cite{su2015fault,su2020byzantine, fang2022bridge,yang2019byrdie}, \emph{similarity-based reweighting (SimRew)} \cite{xu2022byzantine}, \emph{decentralized RSA (DRSA)} \cite{peng2021byzantine} with penalty parameters $0.001$ and $0.5$ in the i.i.d. and non-i.i.d. cases respectively, \emph{centered clipping (CC)} \cite{gorbunov2021secure,karimireddy2021learning} with radii $0.1$ and $0.3$ in the i.i.d. and non-i.i.d. cases respectively, \emph{self centered clipping (SCC)} \cite{he2022byzantine} with radii $0.1$ and $0.3$ in the i.i.d. and non-i.i.d. cases respectively, \emph{FABA}, and \emph{IOS}. For aggregation rules that need $W'$, the mixing matrix of the entire network (including the Byzantine workers), we construct it by the Metropolis-Hastings rule \cite{xiao2004fast}. As a baseline, we also let workers run SGD locally without any communication, marked by \emph{no communication (no comm)}, which is not affected by Byzantine attacks.
	
	We test the performance of the compared algorithms under four popular Byzantine attacks: \emph{Gaussian}, \emph{sign-flipping}\footnote{Sign-flipping is an implementation of \emph{inner-production manipulation} \cite{xie2020fall}.}, \emph{isolation}, \emph{sample-duplicating}, and \emph{a little is enough (ALIE)} \cite{baruch2019little}. For Gaussian, Byzantine worker $b\in \B_n$ generates its message $\tilde{\vx}_{b,n}^{k}$ from a Gaussian distribution with mean $\bar\vx_n^{k}:=(\sum_{m\in\N_n} w_{nm}'\vx^{k}_{m})/(\sum_{m\in\N_n} w_{nm}')$ and variance $1$. For sign-flipping, Byzantine worker $b\in \B_n$ sets its message as $\tilde{\vx}_{b,n}^{k} = -\bar\vx_n^{k}$. For isolation, Byzantine worker $b\in \B_n$ sends $\tilde{\vx}_{b,n}^{k} = (\vx_n^{k} - \sum_{m\in\N_n} w_{nm}'\bar\vx_n^{k})/(\sum_{m\in\B_n} w_{nm}')$ so that the messages received by honest worker $n$ sum up to $x_n^{k}$, which is equivalent to having no communication when weighted mean is used. For sample-duplicating, Byzantine worker $b\in \B_n$ chooses $\tilde{\vx}_{b,n}^{k}$ from $\{\vx^{k}_{m}\}_{m\in\N_n}$ at random.
	
	We first test their performance in the two-castle graph with $10$ honest workers and $2$ Byzantine workers  in Fig. \ref{fig:octopus}. In the numerical experiments, \textit{accuracy (Acc)} and \textit{disagreement measure (DM)} are used as performance measures. Acc stands for the average accuracy of the honest workers on the test images, and DM is defined as \eqref{definition:Hk}. 
	The detailed results for the i.i.d. and non-i.i.d. cases are listed in Tables \ref{table:two-castle-iid} and
	\ref{table:two-castle-non-iid}, respectively.
	According to these experimental results, we have the following conclusions.
	
	\textbf{Most of the existing aggregation rules fail to work well in the decentralized network.} We can observe that CooMed, GeoMed, Krum work well for the i.i.d. case but become unsatisfactory for the non-i.i.d. case. The reasons are that they are essentially designed for i.i.d. data distribution and that they have relatively large contraction constants $\rho$ even in a complete network (cf. Table \ref{table:rho}). TriMean is relatively sensitive to isolation attacks because any honest worker $n$ discards at least $2B_n$ neighboring messages, among which at least $B_n$ are honest. Discarding too many neighboring messages makes honest workers to be isolated easily. The methods based on penalizing differences with neighboring messages like DRSA and clipping neighboring messages with large magnitudes like CC and SCC, are sensitive to sign-flipping attacks because the directions of neighboring workers become critical.
	
	\textbf{Variation is a critical factor to affect robustness.} For the i.i.d. case, most existing aggregation rules achieve acceptable performance. Gaussian and sample-duplicating attacks even do not have any negative influence on some of them. However, their performance degrades sharply for the non-i.i.d. case, because the increased outer variation begins to play a prominent role, which corroborates with our analysis in Theorem \ref{theorem:convergence}.
	
	\begin{figure}[t]
		\centering
		\includegraphics[width=0.8\linewidth]{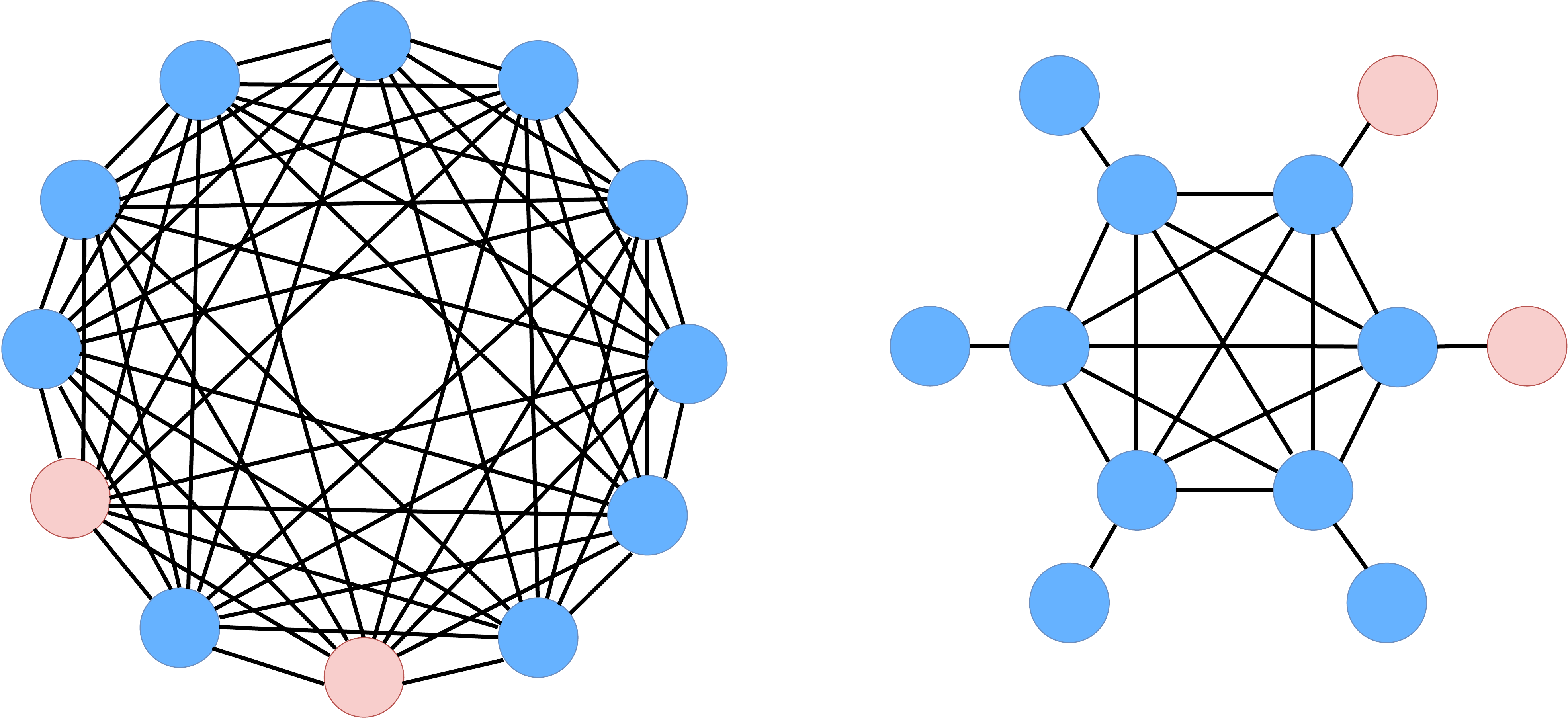}
		\caption{\textbf{(Left)} Two-castle graph and \textbf{(Right)} octopus graph, both consisting of 10 honest and 2 Byzantine workers. The blues and reds represent honest and Byzantine workers, respectively.}
		\label{fig:octopus}
	\end{figure}
	
	\begin{figure}[t]
		\centerline{\includegraphics[width=0.9\columnwidth]{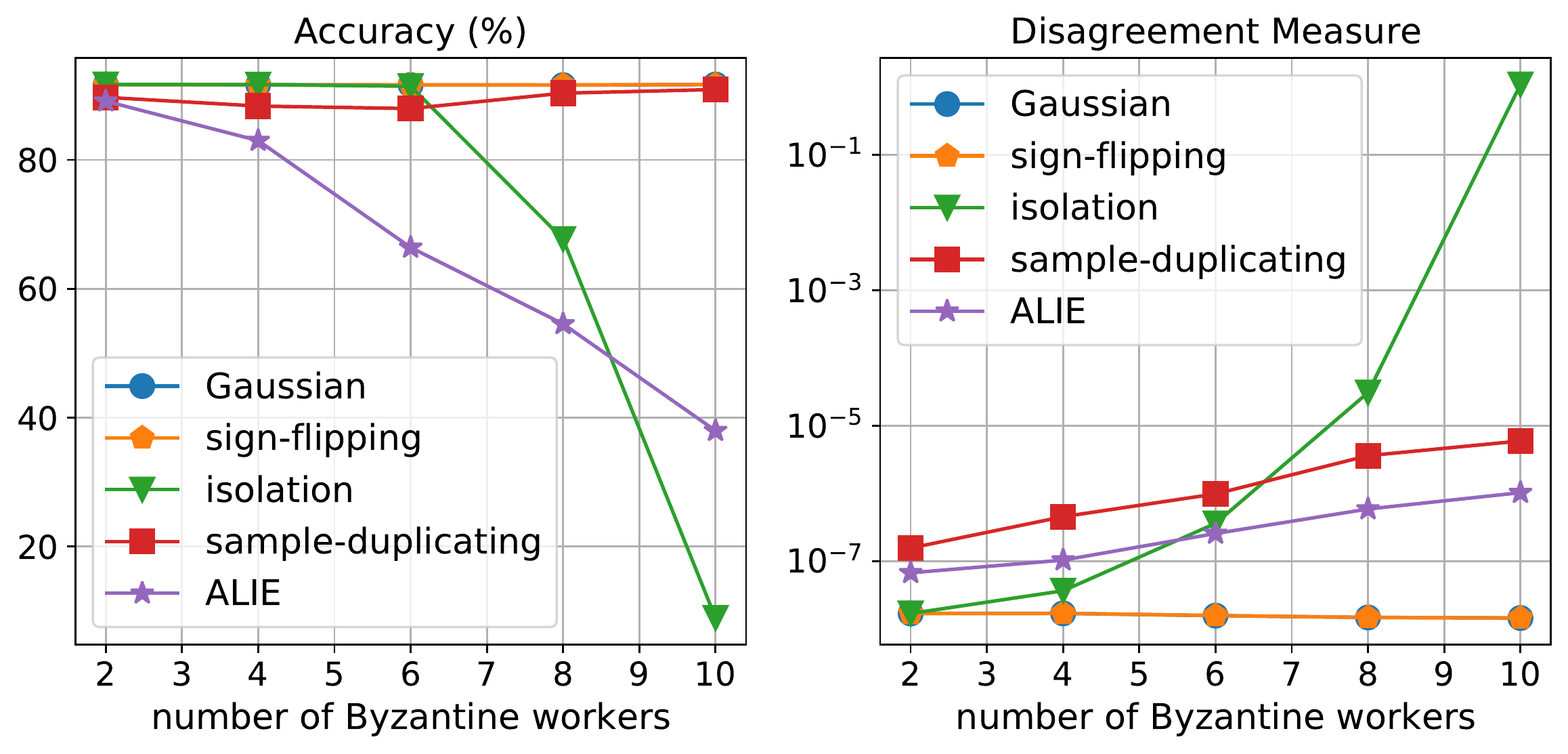}}
		\caption{IOS in Erdos-Renyi graphs under different attacks for the non-i.i.d. case. The number of honest workers is $20$ and the number of Byzantine workers is varying.}
		\label{fig:covergence-IOS-ER}
	\end{figure}
	
	\textbf{Collaboration benefits optimization even in the presence of Byzantine workers.} For the i.i.d. case, honest workers have similar training images and consequently, similar local minima. No collaboration does not matter too much. However, for the non-i.i.d. case, collaboration becomes extremely important. Even when the network is attacked by Byzantine workers, SGD equipped with properly designed robust aggregation rules still outperform SGD without any communication.
	
	\textbf{IOS works well in all numerical experiments.} Observe that IOS achieves similar performance as if the Byzantine workers are absent and WeiMean is used. It outperforms all other aggregation rules, but is very close to FABA. Actually, in this two-castle graph, the virtual mixing matrix of FABA is doubly stochastic too, and thus IOS is equivalent to FABA.
	

	
	
	\textbf{Non-doubly stochastic virtual mixing matrix has negative effect.} To illustrate the influence of non-doubly stochastic virtual mixing matrix, we construct an octopus graph as shown in Fig. \ref{fig:octopus}. The virtual mixing matrix $W$ associated with FABA now becomes non-doubly stochastic if we also choose the weight $w_{nm}'=\frac{1}{N_n+B_n+1}$ for $n \in \N$ and $m \in \N_n \cup \B_n \cup \{n\}$ as we have done in the two-castle graph. Two pairs of algorithms with doubly and non-doubly stochastic virtual mixing matrices are compared: CC/SCC and FABA/IOS. The results listed in Table \ref{table:octopus-non-iid} demonstrate that without the doubly stochastic structures in the virtual mixing matrices, CC and FABA achieve lower accuracy compared with SCC and IOS. This observation validates the theoretical result that the asymptotic learning error depends on $\chi^2$, which describes how non-doubly stochastic $W$ is.
	
	\textbf{IOS is resilient to a large fraction of Byzantine workers.} We also generate larger Erdos-Renyi graphs with $20$ honest workers and varying Byzantine workers. Each pair of workers are neighbors with the probability of $0.7$. We consider non-i.i.d. local data distributions. Fig. \ref{fig:covergence-IOS-ER}
	shows that IOS constantly performs well when the number of Byzantine workers is less than 7, and deteriorates thereafter. These results demonstrate that IOS is resilient to a large fraction of Byzantine workers.
	
	\textbf{More numerical experiments.} To further demonstrate the effectiveness of IOS, we perform more numerical experiments in an Erdos-Renyi graph. The conclusions are consistent to those we have reached above. We leave them to Appendix \ref{section:supplementary-experiments}.

	
	
	\section{Proof of Theorem \ref{theorem:convergence}}
	\label{section:proof-theorem-convergence}
	\allowdisplaybreaks
		In this section, we provide the proof of Theorem \ref{theorem:convergence}.

		By the $L$-smoothness of function $f_n$ in Assumption \ref{assumption:Lip}, $f$ is also $L$-smooth and we have
		\begin{align}
			\label{inequality:convergence-1}
			\E_{\vxi^k} [f(\bar\vx^{k+1})]\le& f(\bar\vx^k)+ \E_{\vxi^k} [\la \nabla f(\bar\vx^k), \bar\vx^{k+1}-\bar\vx^k \ra] \nonumber
			\\
			&+ \frac{L}{2} \E_{\vxi^k} [\|\bar\vx^{k+1}-\bar\vx^k\|^2].
		\end{align}
		With
		the equality $\la \vx, \vy\ra$ $= \frac{1}{2} \|\vx+\vy\|^2-\frac{1}{2}\|\vx\|^2-\frac{1}{2}\|\vy\|^2$, we can rewrite the second term at the RHS of \eqref{inequality:convergence-1} as
		\begin{align}
			\label{inequality:convergence-1-t1}
			&\E_{\vxi^k} [\la \nabla f(\bar\vx^k), \bar\vx^{k+1}-\bar\vx^k \ra]
			\\
			=& \alpha\E_{\vxi^k} [\langle \nabla f(\bar\vx^k), \nabla f(\bar\vx^k; \vxi^k)-\nabla f(\bar\vx^k) + \frac{1}{\alpha}(\bar\vx^{k+1}-\bar\vx^k)
			\rangle]
			\nonumber \\
			=& \frac{\alpha}{2}\E_{\vxi^k} [\| \nabla f(\bar\vx^k; \vxi^k)+ \frac{1}{\alpha} (\bar\vx^{k+1}-\bar\vx^k) \|^2
			- \frac{\alpha}{2}\| \nabla f(\bar\vx^k) \|^2]
			\nonumber \\
			&- \frac{\alpha}{2}\E_{\vxi^k} [\| \nabla f(\bar\vx^k; \vxi^k)-\nabla f(\bar\vx^k)+ \frac{1}{\alpha} (\bar\vx^{k+1}-\bar\vx^k) \|^2].
			\nonumber
		\end{align}
		The third term at the RHS of \eqref{inequality:convergence-1} is bounded by
		\begin{align}
			\label{inequality:convergence-1-t2}
			\hspace{-0.9em}
			&\E_{\vxi^k}[\|\bar\vx^{k+1}\!\!-\!\bar\vx^k\|^2] \!\leq\! 2\alpha^2\E_{\vxi^k}[\|\nabla f(\bar\vx^k; \vxi^k)\!-\!\!\nabla f(\bar\vx^k)\|^2] \\
			\hspace{-0.9em}
			& + 2\alpha^2\E_{\vxi^k}[\| \nabla f(\bar\vx^k; \vxi^k)-\nabla f(\bar\vx^k)+ \frac{1}{\alpha} (\bar\vx^{k+1}-\bar\vx^k) \|^2]
			\nonumber \\
			\hspace{-0.9em}
			\le& \frac{2\alpha^2\delta^2_{\rm in}}{N} \!+\! 2\alpha^2\E_{\vxi^k}[\| \nabla f(\bar\vx^k; \vxi^k)\!-\!\nabla f(\bar\vx^k)\!+\! \frac{1}{\alpha} (\bar\vx^{k+1}\!-\!\bar\vx^k) \|^2],
			\nonumber
		\end{align}
		where the second inequality is from Assumptions \ref{assumption:innerVariance} and \ref{assumption:indSampling}.
		
		As $\alpha \leq \frac{1}{2L}$,
		substituting \eqref{inequality:convergence-1-t1} and \eqref{inequality:convergence-1-t2} into \eqref{inequality:convergence-1} yields
		\begin{align}
			\label{inequality:convergence-2}
			&\E_{\vxi^k}[f(\bar\vx^{k+1})]
			\\
			& \le f(\bar\vx^k)
			+\frac{\alpha}{2}\E_{\vxi^k}[\| \nabla f(\bar\vx^k; \vxi^k)+ \frac{1}{\alpha} (\bar\vx^{k+1}-\bar\vx^k) \|^2]
			\nonumber \\
			& - \frac{\alpha}{2}(1-2\alpha L)\E_{\vxi^k}[\| \nabla f(\bar\vx^k; \vxi^k)-\nabla f(\bar\vx^k)+ \frac{1}{\alpha} (\bar\vx^{k+1}-\bar\vx^k) \|^2]
			\nonumber \\
			&
			- \frac{\alpha}{2}\| \nabla f(\bar\vx^k) \|^2
			+\frac{\alpha^2\delta^2_{\rm in}L}{N}
			\nonumber \\
			& \le f(\bar\vx^k)+\frac{\alpha}{2}\E_{\vxi^k}[\| \nabla f(\bar\vx^k; \vxi^k)+ \frac{1}{\alpha} (\bar\vx^{k+1}-\bar\vx^k) \|^2]
			\nonumber \\
			&
			- \frac{\alpha}{2}\| \nabla f(\bar\vx^k) \|^2
			+\frac{\alpha^2\delta^2_{\rm in}L}{N}. \nonumber
		\end{align}
		
		According to the update rule of \eqref{rule:SGD-robust}, we expand the second term at the RHS of \eqref{inequality:convergence-2} as
		\begin{align}
			\label{equality:global-aggregation-error-1}
			&\nabla f(\bar\vx^k; \vxi^k)+ \frac{1}{\alpha} (\bar\vx^{k+1}-\bar\vx^k)
			\\
			=& \nabla f(\bar\vx^k; \vxi^k)+ \frac{1}{\alpha N}\sum_{n\in\N} (\A_n (\vx^{k+\frac{1}{2}}_n, \{\tilde{\vx}^{k+\frac{1}{2}}_{m,n}\}_{m\in\N_n\cup\B_n})-\bar\vx^k)
			\nonumber \\
			=& \lp \nabla f(\bar\vx^k; \vxi^k) -\frac{1}{N}\sum_{n\in\N} \nabla f_n(\vx^{k}_n; \xi_n^{k}) \rp
			\nonumber \\
			&+\frac{1}{\alpha N}\sum_{n\in\N} \lp\bar\vx^{k+\frac{1}{2}}_n-\bar\vx^k + \frac{\alpha}{N}\sum_{m\in\N} \nabla f_m(\vx^{k}_m; \xi_m^{k})\rp
			\nonumber \\
			&+ \frac{1}{\alpha N}\sum_{n\in\N} \lp\A_n (\vx^{k+\frac{1}{2}}_n, \{\tilde{\vx}^{k+\frac{1}{2}}_{m,n}\}_{m\in\N_n\cup\B_n})-\bar\vx^{k+\frac{1}{2}}_n\rp.
			\nonumber
		\end{align}
		Denote the squared $\ell_2$ norms of the three terms at the RHS of \eqref{equality:global-aggregation-error-1} as $T_1$, $T_2$ and $T_3$, respectively. We establish their upper bounds as follows.
		
		
		\noindent	\textbf{Upper bound of $T_1$.} For $T_1$, it holds that
		\begin{align}
			T_1 := & \E_{\vxi^k}[\|\nabla f(\bar\vx^k; \vxi^k) - \frac{1}{N}\sum_{n\in\N} \nabla f_n(\vx^{k}_n; \xi_n^{k})\|^2]
			\\
			=& \E_{\vxi^k}[\|\frac{1}{N}\sum_{n\in\N}(\nabla f_n(\bar\vx^k; \xi_n^{k}) - \nabla f_n(\vx^{k}_n; \xi_n^{k}))\|^2]
			\nonumber \\
			\le& \frac{1}{N}\sum_{n\in\N}\E_{\xi^k_n}[\|\nabla f_n(\bar\vx^k; \xi_n^{k}) - \nabla f_n(\vx^{k}_n; \xi_n^{k})\|^2]. \nonumber
		\end{align}
		Further, according to Assumption \ref{assumption:Lip}, we have
		\begin{align}
			&\|\nabla f_n(\bar\vx^k; \xi_n^{k}) - \nabla f_n(\vx^{k}_n; \xi_n^{k})\|^2
			\le L^2\|\bar\vx^k-\vx^k_n\|^2.
		\end{align}
		With this inequality, $T_1$ can be bounded by
		\begin{align}
			\label{equality:global-aggregation-error-2-t1}
			T_1 \le \frac{L^2}{N}\sum_{n\in\N}\|\vx^k_n-\bar\vx^k\|^2.
		\end{align}

		\noindent	\textbf{Upper bound of $T_2$.}
		By $\bar\vx^{k+\frac{1}{2}}=\bar\vx^k-\frac{\alpha}{N}\sum_{n\in\N} \nabla f_n(\vx^{k}_n; \xi_n^{k})$,
		we can rewrite the second term at the RHS of \eqref{equality:global-aggregation-error-1} as
		\begin{align}
			\label{equality:global-aggregation-error-1-t3}
			&\frac{1}{\alpha N}\sum_{n\in\N} \lp\bar\vx^{k+\frac{1}{2}}_n-\bar\vx^k
			+\frac{\alpha}{N}\sum_{m\in\N} \nabla f_m(\vx^{k}_m; \xi_m^{k}) \rp
			\\
			=& \frac{1}{\alpha N}\sum_{n\in\N} \lp \bar\vx^{k+\frac{1}{2}}_n-\bar\vx^{k+\frac{1}{2}} \rp. \nonumber
		\end{align}
		Stacking all local models in $X$ as \eqref{definition:stacked-X} and applying Cauchy-Schwarz inequality, we have
		\begin{align}
			\label{equality:global-aggregation-error-2-t2}
			T_2 :=&  \lnorm\frac{1}{\alpha N}\sum_{n\in\N} \lp\bar\vx^{k+\frac{1}{2}}_n-\bar\vx^{k+\frac{1}{2}}\rp\rnorm^2
			\\
			=& \lnorm \frac{1}{\alpha N} \1^\top (WX^{k+\frac{1}{2}}-\frac{1}{N}\1\1^{\top}X^{k+\frac{1}{2}})
			\rnorm^2
			\nonumber \\
			=& \frac{1}{\alpha^2 N^2}\lnorm (\1^\top W-\1^{\top})
			(X^{k+\frac{1}{2}}-\frac{1}{N}\1\1^{\top}X^{k+\frac{1}{2}})
			\rnorm^2
			\nonumber \\
			\le& \frac{1}{\alpha^2 N^2}\lnorm W^\top\1 -\1
			\rnorm^2
			\lnorm X^{k+\frac{1}{2}}-\frac{1}{N}\1\1^{\top}X^{k+\frac{1}{2}}
			\rnorm^2_F
			\nonumber \\
			=& \frac{\chi^2}{\alpha^2 N}\sum_{n\in\N}\lnorm  \vx^{k+\frac{1}{2}}_{n}-\bar\vx^{k+\frac{1}{2}}\rnorm^2. \nonumber
		\end{align}
		
		\noindent \textbf{Upper bound of $T_3$.} From the contraction property of robust aggregation rules $ \{ \A_n \}_{n \in \N}$ in \eqref{inequality:robustness-of-aggregation-local}, $T_3$ can be bounded by
		\begin{align}
			T_3 := &\lnorm\frac{1}{\alpha N}\sum_{n\in\N} \lp\A_n (\vx^{k+\frac{1}{2}}_n, \{\tilde{\vx}^{k+\frac{1}{2}}_{m,n}\}_{m\in\N_n\cup\B_n})-\bar\vx^{k+\frac{1}{2}}_n\rp\rnorm^2
			\nonumber \\
			\le&\frac{1}{\alpha^2 N}\sum_{n\in\N} \lnorm \A_n (\vx^{k+\frac{1}{2}}_n, \{\tilde{\vx}^{k+\frac{1}{2}}_{m,n}\}_{m\in\N_n\cup\B_n})-\bar\vx^{k+\frac{1}{2}}_n\rnorm^2
			\nonumber \\
			\le&\frac{1}{\alpha^2 N}\sum_{n\in\N} \rho^2 \max_{m\in\N_n \cup \{n\}}\| \vx^{k+\frac{1}{2}}_{m}-\bar\vx^{k+\frac{1}{2}}_n\|^2.
		\end{align}
		For any worker $n\in\N$, it holds that
		\begin{align}
			\label{inequality:bridge-global-and-local-avg}
			&\max_{m\in\N_n \cup \{n\}}\|\vx^{k+\frac{1}{2}}_{m}-\bar\vx^{k+\frac{1}{2}}_n\|^2
			\\
			\le& 2\max_{m\in\N_n \cup \{n\}}\|\vx^{k+\frac{1}{2}}_{m}-\bar\vx^{k+\frac{1}{2}}\|^2
			+2\| \bar\vx^{k+\frac{1}{2}}_n - \bar\vx^{k+\frac{1}{2}}\|^2
			\nonumber \\
			\le& 2\max_{n\in\N}\|\vx^{k+\frac{1}{2}}_{n}-\bar\vx^{k+\frac{1}{2}}\|^2
			+2\max_{n\in\N}\|\vx^{k+\frac{1}{2}}_{n}-\bar\vx^{k+\frac{1}{2}}\|^2
			\nonumber \\
			=& 4\max_{n\in\N}\|\vx^{k+\frac{1}{2}}_{n}-\bar\vx^{k+\frac{1}{2}}\|^2,
			\nonumber
		\end{align}
		which implies
		\begin{align}
			\label{equality:global-aggregation-error-2-t3}
			T_3 \le\frac{4\rho^2}{\alpha^2}\max_{n\in\N}\| \vx^{k+\frac{1}{2}}_{n}-\bar\vx^{k+\frac{1}{2}}\|^2.
		\end{align}
		\blue{Note that $\| \bar\vx^{k+\frac{1}{2}}_n - \bar\vx^{k+\frac{1}{2}}\|^2 \leq \max_{n\in\N}\|\vx^{k+\frac{1}{2}}_{n}-\bar\vx^{k+\frac{1}{2}}\|^2$ as the distance between a convex combination of $\{\vx^{k+\frac{1}{2}}_{n}\}_{n \in \N}$ and any point is smaller than the maximum distance between $\{\vx^{k+\frac{1}{2}}_{n}\}_{n \in \N}$ and the point.}
		
		Plugging \eqref{equality:global-aggregation-error-2-t1}, \eqref{equality:global-aggregation-error-2-t2} and \eqref{equality:global-aggregation-error-2-t3} into \eqref{equality:global-aggregation-error-1}, we have
		\begin{align}
			\label{equality:global-aggregation-error-3}
			& \E_{\vxi^k}[\|\nabla f(\bar\vx^k;\vxi^k)+ \frac{1}{\alpha} (\bar\vx^{k+1}-\bar\vx^k)\|^2]
			\\
			\le& 3\E_{\vxi^k}[T_1] + 3\E_{\vxi^k}[T_2] + 3\E_{\vxi^k}[T_3]
			\nonumber \\
			\le& \frac{3L^2}{N}\sum_{n\in\N}\|\vx^k_n-\bar\vx^k\|^2+\frac{3\chi^2}{\alpha^2 N}\E_{\vxi^k} [\sum_{n\in\N}\| \vx^{k+\frac{1}{2}}_{n}-\bar\vx^{k+\frac{1}{2}}\|^2]
			\nonumber \\
			&+ \frac{12\rho^2}{\alpha^2} \E_{\vxi^k}[\max_{n\in\N} \| \vx^{k+\frac{1}{2}}_{n}-\bar\vx^{k+\frac{1}{2}}\|^2].
			\nonumber
		\end{align}
		Plugging \eqref{inequality:ce-1/2-iteration-I} with $v=\frac{1}{2}$ in Lemma \ref{lemma:ce-1/2-iteration} into \eqref{equality:global-aggregation-error-3} yields
		\begin{align}
			\label{equality:global-aggregation-error-4}
			&\E_{\vxi^k}[\|\nabla f(\bar\vx^k)+ \frac{1}{\alpha} (\bar\vx^{k+1}-\bar\vx^k)\|^2]
			\\
			\le& \frac{36\rho^2}{\alpha^2} \max_{n\in\N} \| \vx^{k}_{n}-\bar\vx^{k}\|^2
			+(\frac{9\chi^2}{\alpha^2N}+\frac{3L^2}{N})\sum_{n\in\N} \| \vx^{k}_{n}-\bar\vx^{k}\|^2
			\nonumber \\
			&+24(4\rho^2N +\chi^2)(\delta^2_{\rm in}+\delta^2_{\rm out})
			\nonumber \\
			\le& (\frac{36\rho^2N}{\alpha^2} +\frac{9\chi^2}{\alpha^2}+3L^2)
			H^k
			+24(4\rho^2N +\chi^2)(\delta^2_{\rm in}+\delta^2_{\rm out}).
			\nonumber
		\end{align}
		
		Reorganizing the terms in \eqref{inequality:convergence-2} and  substituting \eqref{equality:global-aggregation-error-4} lead to
		\begin{align}
			\label{inequality:convergence-3}
			& \|\nabla f(\bar\vx^k)\|^2
			\\
			\le&\frac{2\E_{\vxi^k}[f(\bar\vx^k)-f(\bar\vx^{k+1})]}{\alpha}
			+ \frac{2\alpha\delta^2_{\rm in}L}{N}
			\nonumber \\
			&+ \E_{\vxi^k}[\|\nabla f(\bar\vx^k)+ \frac{1}{\alpha} (\bar\vx^{k+1}-\bar\vx^k) \|^2]
			\nonumber \\
			\le& \frac{2\E_{\vxi^k}[f(\bar\vx^k)-f(\bar\vx^{k+1})]}{\alpha}
			+ \frac{2\alpha\delta^2_{\rm in}L}{N}
			\nonumber \\
			&+(\frac{36\rho^2N}{\alpha^2} +\frac{9\chi^2}{\alpha^2}+3L^2)H^k
			+24(4\rho^2N +\chi^2)(\delta^2_{\rm in}+\delta^2_{\rm out}).
			\nonumber
		\end{align}
		Taking expectation and averaging \eqref{inequality:convergence-3} over $k=1, \cdots, K$ give
		\begin{align}
			\label{inequality:convergence-4}
			&\frac{1}{K}\sum_{k=1}^{K}\E[\|\nabla f(\bar\vx^k)\|^2]
			\\
			\hspace{-1em} \le& \frac{2\E[f(\bar\vx^0)-f(\bar\vx^{k+1})]}{\alpha K}
			+ \frac{2\alpha\delta^2_{\rm in}L}{N}
			+24(4\rho^2N +\chi^2)(\delta^2_{\rm in}+\delta^2_{\rm out})
			\nonumber \\
			\hspace{-1em} &+\lp 36\rho^2N + 9\chi^2+3\alpha^2L^2\rp
			\frac{1}{\alpha^2 K}\sum_{k=1}^{K} \E [H^k]
			\nonumber \\
			\hspace{-1em} \le& \frac{2(f(\bar\vx^0)-f^*)}{\alpha K}
			+ \frac{2\alpha\delta^2_{\rm in}L}{N}
			+96(\rho^2N +\chi^2)(\delta^2_{\rm in}+\delta^2_{\rm out})
			\nonumber \\
			\hspace{-1em} &+ (36(\rho^2N + \chi^2)+3\alpha^2L^2)\
			\frac{1}{\alpha^2 K}\sum_{k=1}^{K} \E [H^k],
			\nonumber
		\end{align}
		which completes the proof.

	\section{Conclusions}
	
	This paper deals with the Byzantine-resilient decentralized stochastic optimization problem. We reveal the intrinsic challenges arising from the distributed scenario to decentralized, and give design guidelines of developing provably Byzantine-resilient algorithms. Following these guidelines, we devise a novel set of robust aggregation rules, IOS, and demonstrate its superior performance with numerical experiments.
	
	\balance

	\bibliographystyle{unsrt}
	\bibliography{abrv,Byzantine,decentralize,federated,textbook}

	\appendices
	\section{Proof of Theorem \ref{theorem:consensus}}
	\label{section:proof-theorem-consensus}
	
	
	Before bounding the disagreement measure,
	we introduce its matrix form.
	Following the same notation of \eqref{definition:stacked-X} in Section \ref{section:challenge-of-robust-aggregations}, $(\frac{1}{N}\1^\top X^k)^\top = \frac{1}{N}\sum_{n\in\N}\vx^k_n=\bar\vx^k$ is the average of all honest models, and $X^k-\frac{1}{N}\1\1^\top X^k = (I-\frac{1}{N}\1\1^\top)X^k$ is the stacked disagreement matrix whose norms represent the disagreement among the honest workers. With these notations, we can write the disagreement measure $H^k$ in \eqref{definition:Hk} as
	\begin{align}
		H^k = \frac{1}{N}\sum_{n\in\N}\|\vx^k_n - \bar\vx^k\|^2 = \frac{1}{N}\|(I-\frac{1}{N}\1\1^\top)X^k\|_{F}^2,
	\end{align}
	where $\|\cdot\|_F$ is the Frobenius norm.

	With these preparations, we begin our proof of Theorem \ref{theorem:consensus}.
	\begin{proof}
		For any $u \in (0, 1)$, it holds that
		\begin{align}
			\label{inequality:ce-iteration-0}
			& \|(I-\frac{1}{N}\1\1^\top) X^{k+1}\|_F^2
			\\
			\hspace{-1em} \le & \frac{1}{1-u}\|(I-\frac{1}{N}\1\1^\top)W X^{k+\frac{1}{2}}\|_F^2
			+ \frac{2}{u}\|X^{k+1}-WX^{k+\frac{1}{2}}\|_F^2
			\nonumber \\
			&
			+ \frac{2}{u}\|\frac{1}{N}\1\1^\top X^{k+1}-\frac{1}{N}\1\1^\top WX^{k+\frac{1}{2}}\|_F^2, \nonumber
		\end{align}
		where the inequality comes from $\|\vx+\vy+\vz\|_F^2\le \frac{1}{1-u}\|\vx\|_F^2+\frac{2}{u}\|\vy\|_F^2+\frac{2}{u}\|\vz\|_F^2$.
		
		Since $W$ is row stochastic, it holds that $W\1 = \1$, with which the first term at the RHS of \eqref{inequality:ce-iteration-0} can be bounded by
		\begin{align}
			\label{inequality:ce-iteration-0-1}
			&\|(I-\frac{1}{N}\1\1^\top) WX^{k+\frac{1}{2}}\|_F^2
			\\
			=& \|(I-\frac{1}{N}\1\1^\top)W (I-\frac{1}{N}\1\1^\top) X^{k+\frac{1}{2}}\|_F^2
			\nonumber \\
			\le& \|(I-\frac{1}{N}\1\1^\top)W \|^2\|(I-\frac{1}{N}\1\1^\top) X^{k+\frac{1}{2}}\|_F^2
			\nonumber \\
			=& (1-\lambda) \| (I-\frac{1}{N}\1\1^\top) X^{k+\frac{1}{2}}\|_F^2. \nonumber
		\end{align}
		
		With the contraction property of robust aggregation rules $ \{ \A_n \}_{n \in \N}$ in \eqref{inequality:robustness-of-aggregation-local}, we can bound the second term at the RHS of \eqref{inequality:ce-iteration-0} by
		\begin{align}
			\label{inequality:ce-iteration-0-2}
			&\|X^{k+1}-WX^{k+\frac{1}{2}}\|_F^2
			\\
			=& \sum_{n\in\N} \|\A_n (\vx^{k+\frac{1}{2}}_n, \{\tilde{\vx}^{k+\frac{1}{2}}_{m,n}\}_{m\in\N_n\cup\B_n})-\bar\vx^{k+\frac{1}{2}}_n\|^2
			\nonumber \\
			\le& \rho^2\sum_{n\in\N}\max_{m\in\N_n \cup \{n\}}\| \vx^{k+\frac{1}{2}}_{m}-\bar\vx^{k+\frac{1}{2}}_n\|^2
			\nonumber \\
			\le& 4\rho^2N\max_{n\in\N}\| \vx^{k+\frac{1}{2}}_{n}-\bar\vx^{k+\frac{1}{2}}\|^2
			\nonumber \\
			\le& 4\rho^2N\sum_{n\in\N}\| \vx^{k+\frac{1}{2}}_{n}-\bar\vx^{k+\frac{1}{2}}\|^2
			\nonumber \\
			\le & 4\rho^2N\|(I-\frac{1}{N}\1\1^\top)X^{k+\frac{1}{2}}\|_F^2, \nonumber
		\end{align}
		\blue{where the second inequality comes from \eqref{inequality:bridge-global-and-local-avg}.}
		
		For the third term at the RHS of \eqref{inequality:ce-iteration-0}, it holds that
		\begin{align}
			\label{inequality:ce-iteration-0-3}
			&\|\frac{1}{N}\1\1^\top X^{k+1}-\frac{1}{N}\1\1^\top WX^{k+\frac{1}{2}}\|_F^2
			\\
			\le&
			\|\frac{1}{N}\1\1^\top\|^2
			\|X^{k+1}- WX^{k+\frac{1}{2}}\|^2_F
			\nonumber \\
			\le&
			4\rho^2N\|(I-\frac{1}{N}\1\1^\top)X^{k+\frac{1}{2}}\|_F^2,
			\nonumber
		\end{align}
		\blue{where the second inequality uses $\|\frac{1}{N}\1\1^\top\|^2=1$ and \eqref{inequality:ce-iteration-0-2}.}
		
		Substituting \eqref{inequality:ce-iteration-0-1}--\eqref{inequality:ce-iteration-0-3} back into \eqref{inequality:ce-iteration-0}, we have
		\begin{align}
			&\|(I-\frac{1}{N}\1\1^\top) X^{k+1}\|_F^2
			\\
			\le &\lp\frac{1-\lambda}{1-u}+ \frac{16\rho^2N}{u}\rp\|(I-\frac{1}{N}\1\1^\top)X^{k+\frac{1}{2}}\|_F^2.
			\nonumber
		\end{align}
		Taking expectation over $\vxi^k$ and applying Lemma \ref{lemma:ce-1/2-iteration}, we have
		\begin{align}
			\label{inequality:ce-iteration-1}
			&\E_{\vxi^k}[H^{k+1}] \\
			\le& \lp\frac{1-\lambda}{1-u} + \frac{16\rho^2N}{u}\rp
			\lp\frac{1}{1-v}+\frac{6\alpha^2L^2}{v}\rp
			H^k
			\nonumber \\
			&+ \lp\frac{1-\lambda}{1-u} + \frac{16\rho^2N}{u}\rp \frac{4\alpha^2}{v} (\delta^2_{\rm in} +\delta^2_{\rm out})
			\nonumber \\
			\le& (1-\lambda+8\rho \sqrt{N})\lp\frac{1}{1-v}+\frac{6\alpha^2L^2}{v}\rp
			H^k
			\nonumber \\
			&+ (1-\lambda+8\rho \sqrt{N}) \frac{4\alpha^2}{v}  (\delta^2_{\rm in} +\delta^2_{\rm out})
			\nonumber \\
			\le& (1-\omega)\lp\frac{1}{1-v}+\frac{6\alpha^2L^2}{v}\rp
			H^k + (1-\omega) \frac{4\alpha^2}{v} (\delta^2_{\rm in} +\delta^2_{\rm out}),
			\nonumber
		\end{align}
		where the second inequality chooses $u=4\rho\sqrt{N} \le \lambda$ and uses the inequality $\frac{1-\lambda}{1-u}\le 1-\lambda+u$.
		To bound the coefficient of the first term at the RHS of \eqref{inequality:ce-iteration-1}, we set $v=\frac{\omega}{3}$ and observe that the step size rule implies
		\begin{align}
			6\alpha^2L^2 \le& \frac{(2-\omega)\omega^2}{3(3-\omega)}
			= \frac{(2\omega/3-\omega^2/3)}{1-\omega/3}\cdot\frac{\omega}{3}
			\\
			=& \frac{(\omega-v-v\omega)v}{1-v}, \nonumber
		\end{align}
		with which we know that
		\begin{align}
			\label{inequality:ce-iteration-5}
			\frac{1}{1-v}+\frac{6\alpha^2L^2}{v}
			\le \frac{1+\omega-v-v\omega}{1-v}
			= 1+\omega.
		\end{align}
		
		Substituting \eqref{inequality:ce-iteration-5} back into \eqref{inequality:ce-iteration-1} yields
		\begin{align}
			\label{inequality:ce-iteration-6-Hk-iteration}
			\E_{\vxi^k}[H^{k+1}]
			\le & \lp 1-\omega^2\rp H^k
			+ \frac{12(1-\omega)}{\omega} \alpha^2(\delta^2_{\rm in} +\delta^2_{\rm out}).
		\end{align}
		Using telescopic cancellation on \eqref{inequality:ce-iteration-6-Hk-iteration} from $0$ to $k$, we deduce that
		\begin{align}
			\label{inequality:ce-Hkp-convergence-1}
			\E [H^{k}]
			\le & (1-\omega^2)^k \lp H^{0}-\frac{12(1-\omega)}{\omega^3} \alpha^2(\delta^2_{\rm in} +\delta^2_{\rm out})\rp
			\notag\\
			&+ \frac{12(1-\omega)}{\omega^3} \alpha^2(\delta^2_{\rm in} +\delta^2_{\rm out}).
		\end{align}
		Since for all honest workers $n \in \N$, $\vx^0_n$ are initialized at the same point, it holds that $H^{0}=0$ and
		\begin{align}
			\label{inequality:ce-Hkp-convergence}
			\E [H^{k}]
			\le \alpha^2\Delta(\delta^2_{\rm in} +\delta^2_{\rm out}),
		\end{align}
		which completes the proof.
	\end{proof}

	\section{Proof of Theorem \ref{theorem:robustness-of-aggregation}}
	\begin{proof}
		Since $W'$ is a doubly stochastic matrix, we know $W$ is doubly stochastic as well. Now we prove \eqref{equality:rho-of-ios}.
		
		At any inner iteration $i$ in Algorithm \ref{algorithm:ios}, the removed model $\vx_{m^{(i)}}$ can be either Byzantine or honest for any honest worker $n$.
		Let us define $\N^{(i)}_n:=\N\cap\cU^{(i)}_n$ and $\B^{(i)}_n:=\B\cap\cU^{(i)}_n$. Also define $\bar\vx^{(i)}_N := \frac{1}{\green{\ccalW'_n}(\N^{(i)}_n)}\sum_{m\in \N^{(i)}_n} w_{nm}'\vx_m$ and $\bar\vx^{(i)}_B := \frac{1}{\green{\ccalW'_n}(\B^{(i)}_n)}\sum_{b\in\B^{(i)}_n} w_{nb}'\vx_b$. For simplicity, we omit subscript $n$ in the notations of $\bar\vx^{(i)}_N$, $\bar\vx^{(i)}_B$ and $\vx^{(i)}_{\text{avg}}$. Observe that
		\begin{align}
			\vx^{(i)}_{\text{avg}}=&\frac{\sum_{m\in \mathcal{U}_n^{(i)}}w_{nm}'\vx_m}{\sum_{m\in \mathcal{U}_n^{(i)}}w_{nm}'}
			\\
			=& \frac{\sum_{m\in \N^{(i)}_n}w_{nm}'\vx_m+\sum_{b\in \B^{(i)}_n}w_{nb}'\vx_b}{\sum_{m\in \mathcal{U}_n^{(i)}}w_{nm}'}
			\nonumber \\
			=& (1-\mu^{(i)}_n)\bar\vx^{(i)}_N + \mu^{(i)}_n\bar\vx^{(i)}_B, \nonumber
		\end{align}
		where $\mu^{(i)}_n\!:=\!\frac{\!\sum_{b\in \B^{(i)}_n}\!w_{nb}'}{\!\sum_{m\in \mathcal{U}_n^{(i)}}\!w_{nm}'}$ is the normalized weight from $\B^{(i)}_n$.
		
		\textbf{Case 1:} When inequality $\|\bar\vx^{(i)}_N-\bar\vx^{(i)}_B\| > \max_{m\in\N_n}\|\vx_m - \bar\vx^{(i)}_N\|$ $/(1-2\mu^{(i)}_n)$ is satisfied, it holds that
		\begin{align}
			\label{inequality:robustness-of-IOS}
			& \mu^{(i)}_n \|\bar\vx^{(i)}_N-\bar\vx^{(i)}_B\| + \max_{m\in\N_n}\|\vx_m - \bar\vx^{(i)}_N\| \\
			< & (1-\mu^{(i)}_n) \|\bar\vx^{(i)}_N-\bar\vx^{(i)}_B\|. \nonumber
		\end{align}
		Observe that there exists at least one Byzantine neighbor $b\in\B^{(i)}_n$ such that
		\begin{align}
			\label{new-eq-001}
			\|\vx_b-\vx^{(i)}_{\text{avg}}\|
			\ge & \|\bar\vx^{(i)}_B-\vx^{(i)}_{\text{avg}}\| \\
			=   & (1-\mu^{(i)}_n) \|\bar\vx^{(i)}_N-\bar\vx^{(i)}_B\|,\nonumber
		\end{align}
		and for any honest neighbor $n'\in \N^{(i)}_n$ it holds that
		\begin{align}
			\label{new-eq-002}
			\|\vx_{n'}-\vx^{(i)}_{\text{avg}}\|
			\le& \|\vx_{n'} - \bar\vx^{(i)}_N\| + \|\bar\vx^{(i)}_N-\vx^{(i)}_{\text{avg}}\|
			\\
			\le & \max_{m\in\N_n}\|\vx_m - \bar\vx^{(i)}_N\| + \mu^{(i)}_n\|\bar\vx^{(i)}_N-\bar\vx^{(i)}_B\|.
			\nonumber
		\end{align}
		From \eqref{inequality:robustness-of-IOS}, \eqref{new-eq-001} and \eqref{new-eq-002}, there exists at least one Byzantine neighbor $b\in\B^{(i)}_n$ such that $\|\vx_b-\vx^{(i)}_{\text{avg}}\|>\|\vx_{n'}-\vx^{(i)}_{\text{avg}}\|$ holds for any honest neighbor $n'\in \N^{(i)}_n$. Thus, we can successfully remove one Byzantine neighbor in inner iteration $i$. This ends the discussion of Case 1.
		
		
		\textbf{Case 2:} When inequality $\|\bar\vx^{(i)}_N-\bar\vx^{(i)}_B\| \le \max_{m\in\N_n}\|\vx_m - \bar\vx^{(i)}_N\|$ $/(1-2\mu^{(i)}_n)$ is satisfied, it becomes possible to discard the model from a honest neighbor by mistake. However, it still holds that
		\begin{align}
			\label{inequality:robustness-of-IOS-1}
			&\|\vx^{(i+1)}_{\text{avg}}-\bar\vx_n\|
			\\
			\le& \|\vx^{(i+1)}_{\text{avg}}-\vx^{(i)}_{\text{avg}}\| + \|\vx^{(i)}_{\text{avg}}-\bar\vx^{(i)}_N\| + \|\bar\vx^{(i)}_N-\bar\vx_n\|
			\nonumber \\
			\le& \frac{w_{nm^{(i)}}'}{\green{\ccalW'_n}(\cU_n^{(i+1)})}\|\vx_{m^{(i)}} - \vx^{(i)}_{\text{avg}}\| + \|\vx^{(i)}_{\text{avg}}-\bar\vx^{(i)}_N\| + \|\bar\vx^{(i)}_N-\bar\vx_n\|
			\nonumber \\
			\le& \frac{w_{nm^{(i)}}'}{\green{\ccalW'_n}(\cU_n^{(i+1)})}\|\vx_{m^{(i)}} - \bar\vx^{(i)}_N\|
			\nonumber \\
			&+ \lp1+\frac{w_{nm^{(i)}}'}{\green{\ccalW'_n}(\cU_n^{(i+1)})}\rp\|\vx^{(i)}_{\text{avg}}-\bar\vx^{(i)}_N\| + \|\bar\vx^{(i)}_N-\bar\vx_n\|
			\nonumber \\
			\le& \frac{w_{nm^{(i)}}'}{\green{\ccalW'_n}(\cU_n^{(i+1)})}\|\vx_{m^{(i)}} - \bar\vx^{(i)}_N\|
			+ 2\|\vx^{(i)}_{\text{avg}}-\bar\vx^{(i)}_N\| + \|\bar\vx^{(i)}_N-\bar\vx_n\|,
			\nonumber
		\end{align}
		in which the second inequality comes from inequality \eqref{inequality:different-partial-average-1} in Lemma \ref{lemma:different-partial-average} and the fourth inequality holds because
		\begin{align}
			\label{new-eq-003}
			\green{\ccalW'_n}(\cU_n^{(i+1)}) \ge 1-\green{\ccalW'_n}(\green{\cU^{\rm max}_n}) \ge \green{\ccalW'_n}(\green{\cU^{\rm max}_n}) \ge w_{nm^{(i)}}'.
		\end{align}
		Note that $1-\green{\ccalW'_n}(\green{\cU^{\rm max}_n}) \ge \green{\ccalW'_n}(\green{\cU^{\rm max}_n})$ as we require $\green{\ccalW'_n}(\green{\cU^{\rm max}_n}) < \frac{1}{3}$.
		
		The hypothesis on $\|\bar\vx^{(i)}_N-\bar\vx^{(i)}_B\|$ guarantees that
		\begin{align}
			\label{inequality:robustness-of-IOS-1-2}
			\|\vx^{(i)}_{\text{avg}}-\bar\vx^{(i)}_N\|
			= &\mu^{(i)}_n\|\bar\vx^{(i)}_N-\bar\vx^{(i)}_B\|
			\\
			\le& \frac{\mu^{(i)}_n}{1-2\mu^{(i)}_n} \max_{m\in\N_n}\|\vx_m - \bar\vx^{(i)}_N\|. \nonumber
		\end{align}
		From \eqref{inequality:different-partial-average-1} in Lemma \ref{lemma:different-partial-average}, we have
		\begin{align}
			\label{inequality:robustness-of-IOS-1-3}
			&\|\bar\vx^{(i)}_N-\bar\vx_n\|
			\\
			\le& \frac{\green{\ccalW'_n}(\N_n)+w'_{nn}+\green{\ccalW'_n}(\B_n)-\green{\ccalW'_n}(\N^{(i)}_n)}{\green{\ccalW'_n}(\N_n)+w_{nn}'+\green{\ccalW'_n}(\B_n)}\max_{m\in\N_n}\|\vx_m - \bar\vx^{(i)}_N\|. \nonumber
		\end{align}
		Substituting \eqref{inequality:robustness-of-IOS-1-2} and \eqref{inequality:robustness-of-IOS-1-3} into \eqref{inequality:robustness-of-IOS-1} yields
		\begin{align}
			\label{inequality:robustness-of-IOS-2}
			&\|\vx^{(i+1)}_{\text{avg}}-\bar\vx_n\|
			\\
			\le &\lp\frac{w_{nm^{(i)}}'}{\green{\ccalW'_n}(\cU^{(i+1)}_n)}
			+ \frac{\green{\ccalW'_n}(\N_n)+w'_{nn}+\green{\ccalW'_n}(\B_n)-\green{\ccalW'_n}(\N^{(i)}_n)}{\green{\ccalW'_n}(\N_n)+w'_{nn}+\green{\ccalW'_n}(\B_n)}
			\right.
			\nonumber \\
			&\left.+ \frac{2\mu^{(i)}_n}{1-2\mu^{(i)}_n} \rp
			\cdot \max_{m\in\N_n} \|\vx_m - \bar\vx^{(i)}_N\|
			. \nonumber
		\end{align}
		
		Now we bound the coefficient at the RHS of \eqref{inequality:robustness-of-IOS-2}. Since $(\N_n\cup\{n\})\setminus\N^{(i)}_n$ is the set of the discarded honest neighbors up to iteration $i$, we know that $m^{(i)}\notin (\N_n\cup\{n\})\setminus\N^{(i)}_n$. In addition, $\{m^{(i)}\}\cup((\N_n\cup\{n\})\setminus\N^{(i)}_n)$ contains no more than $i+1\le q_n$ elements, implying
		\begin{align}
			\label{new-eq-004}
			w_{nm^{(i)}}'+\green{\ccalW'_n}(\N_n)+w'_{nn}-\green{\ccalW'_n}(\N^{(i)}_n)\le \green{\ccalW'_n}(\green{\cU^{\rm max}_n}).
		\end{align}
		Observe the following relations $\green{\ccalW'_n}(\cU^{(i+1)}_n)\ge \green{\ccalW'_n}(\N^{(i+1)}_n)$, $\green{\ccalW'_n}(\N_n)+w'_{nn}+\green{\ccalW'_n}(\B_n)\ge \green{\ccalW'_n}(\N^{(i+1)}_n)$ and $\green{\ccalW'_n}(\N^{(i+1)}_n)\ge 1-2\green{\ccalW'_n}(\green{\cU^{\rm max}_n})$ all hold. Along with \eqref{new-eq-003}, we can bound the first and the third coefficients at the RHS of \eqref{inequality:robustness-of-IOS-2} by
		\begin{align}
			\label{inequality:robustness-of-IOS-2-1&3}
			&\frac{w_{nm^{(i)}}'}{\green{\ccalW'_n}(\cU^{(i+1)}_n)} \!+\! \frac{\green{\ccalW'_n}(\N_n)+w'_{nn}+\green{\ccalW'_n}(\B_n)-\green{\ccalW'_n}(\N^{(i)}_n)}{\green{\ccalW'_n}(\N_n)+w'_{nn}+\green{\ccalW'_n}(\B_n)}
			\\
			\le& \frac{w_{nm^{(i)}}'+\green{\ccalW'_n}(\N_n)+w'_{nn}+\green{\ccalW'_n}(\B_n)-\green{\ccalW'_n}(\N^{(i)}_n)}{\green{\ccalW'_n}(\N^{(i+1)}_n)}
			\nonumber\\
			\le& \frac{2\green{\ccalW'_n}(\green{\cU^{\rm max}_n})}{1-2\green{\ccalW'_n}(\green{\cU^{\rm max}_n})}.
			\nonumber
		\end{align}
		In addition, $\mu^{(i)}_n$ contains the weights of Byzantine neighbors but can be bounded by
		\begin{align}
			\label{inequality:robustness-of-IOS-2-2}
			\mu^{(i)}_n=\frac{\green{\ccalW'_n}(\B_n^{(i)})}{\green{\ccalW'_n}(\cU^{(i)}_n)}\le \frac{\green{\ccalW'_n}(\B_n)}{1-\green{\ccalW'_n}(\green{\cU^{\rm max}_n})}\le \frac{\green{\ccalW'_n}(\green{\cU^{\rm max}_n})}{1-\green{\ccalW'_n}(\green{\cU^{\rm max}_n})}.
		\end{align}
		Plugging \eqref{inequality:robustness-of-IOS-2-1&3} and \eqref{inequality:robustness-of-IOS-2-2} into \eqref{inequality:robustness-of-IOS-2} leads to
		\begin{align}
			\label{inequality:robustness-of-IOS-3}
			\|\vx^{(i+1)}_{\text{avg}}-\bar\vx_n\|
			\le& \frac{4\green{\ccalW'_n}(\green{\cU^{\rm max}_n})}{1-3\green{\ccalW'_n}(\green{\cU^{\rm max}_n})} \max_{m\in\N_n} \|\vx_m - \bar\vx^{(i)}_N\|
			\\
			\le& \frac{12\green{\ccalW'_n}(\green{\cU^{\rm max}_n})}{1-3\green{\ccalW'_n}(\green{\cU^{\rm max}_n})} \max_{m\in\N_n} \|\vx_m - \bar\vx_n\|
			. \nonumber
		\end{align}
		To derive the last inequality of \eqref{inequality:robustness-of-IOS-3}, we use \eqref{inequality:different-partial-average-2} in Lemma \ref{lemma:different-partial-average}.
		With it, we have
		\begin{align}
			\label{inequality:robustness-of-IOS-3-1}
			&\max_{m\in\N_n} \|\vx_m - \bar\vx^{(i)}_N\|
			\\
			\le& \max_{m\in\N_n} \|\vx_m - \bar\vx_n\| + \max_{m\in\N_n} \|\bar\vx_n - \bar\vx^{(i)}_N\|
			\nonumber \\
			\le& \lp 1+\frac{\green{\ccalW'_n}(\N_n)+w_{nn}'+\green{\ccalW'_n}(\B_n)-\green{\ccalW'_n}(\N^{(i)}_n)}{\green{\ccalW'_n}(\N^{(i)}_n)}\rp
			\nonumber \\
			&\cdot \max_{m\in\N_n} \|\vx_m - \bar\vx_n\|
			\nonumber \\
			\le& \frac{1}{1-2\green{\ccalW'_n}(\green{\cU^{\rm max}_n})} \max_{m\in\N_n} \|\vx_m - \bar\vx_n\|
			\nonumber \\
			\le& 3 \max_{m\in\N_n} \|\vx_m - \bar\vx_n\|. \nonumber
		\end{align}
		This ends the discussion of Case 2.
		
		Therefore, we conclude that for
		\begin{align}
			\A_n (\vx_n, \{\tilde{\vx}_{m,n}\}_{m\in\N_n\cup\B_n})=\vx^{(q_n)}_{\text{avg}},
		\end{align}
		no matter which case happens at each inner iteration $i=0$, $\cdots$, $q_n-1$, we eventually have
		\begin{align}
			\|\A_n (\vx_n, & \{\tilde{\vx}_{m,n}\}_{m\in\N_n\cup\B_n})-\bar\vx_n \|
			\\
			\le& \frac{12\green{\ccalW'_n}(\green{\cU^{\rm max}_n})}{1-3\green{\ccalW'_n}(\green{\cU^{\rm max}_n})} \max_{m\in \N_n}\|\vx_m - \bar\vx_n\|.
			\nonumber
		\end{align}
		According to Definition \ref{definition:mixing-matrix}, the contraction constant $\rho$ is bounded by
		\begin{align}
			\rho \le \max_{n\in\N}\frac{12\green{\ccalW'_n}(\green{\cU^{\rm max}_n})}{1-3\green{\ccalW'_n}(\green{\cU^{\rm max}_n})},
		\end{align}
		which completes the proof.
	\end{proof}

	\section{Coverage of Generic Aggregator Form}
	\label{section:coverage-of-generic-form}
	We next show that the existing Byzantine-resilient decentralized (stochastic)   algorithms all fall in the   form of \eqref{definition:aggregation-centralized-to-decentralized}.
	
	The works of \cite{su2015fault,su2020byzantine,fang2022bridge,yang2019byrdie} adopt \textit{trimmed mean (TriMean)} as the base aggregator, given by
	\begin{align}
		\hspace{-0.7em}
		\A(\vx_n,\! \{\tilde{\vx}_{m,n} \}_{m\in\N_n\cup\B_n})\! =\! \operatorname{TriMean}(\{\tilde{\vx}_{m,n} \}_{m\in\N_n\cup\B_n}),
	\end{align}
	and set $r_n=\frac{1}{N_n+B_n-2q+1}$ where $q$ is a parameter. For each coordinate $d=1, \ldots, D$, trimmed mean discards the largest $q$ and the smallest $q$ elements, and returns the average of the remaining ones.
	
	
	
	
	The work of \cite{peng2021byzantine} proposes \textit{decentralized RSA}, in which
	\begin{align}
		&\A_n (\vx_n, \{\tilde{\vx}_{m,n}\}_{m\in\N_n\cup\B_n})
		\\
		=& \vx_n + \alpha^k C_{R} \!\!\!\!\!\! \sum_{m\in \N_n\cup\B_n} \operatorname{Sign}(\tilde{\vx}_{m,n}-\vx_n)
		\nonumber\\
		=& (1-\alpha^k C_{R})\vx_n + \alpha^k C_{R} \lp \vx_n +\!\!\!\!\!\!\sum_{m\in \N_n\cup\B_n} \operatorname{Sign}(\tilde{\vx}_{m,n}-\vx_n)\rp.
		\nonumber
	\end{align}
	Therein, $\operatorname{Sign}$ denotes the element-wise sign function, $\alpha^k > 0$ is the step size, and $C_{R}$ is a constant. We can observe that the base aggregator $\A$ is
	\begin{align}
		\A(\vx_n, \{\tilde{\vx}_{m,n}\}_{m\in\N_n\cup\B_n})
		=  \vx_n + \!\!\!\!\!\! \sum_{m\in \N_n\cup\B_n} \operatorname{Sign}(\tilde{\vx}_{m,n}-\vx_n),
	\end{align}
	and $r_n=1-\alpha^k C_{R}$, which can be time-varying.
	
	The work of \cite{he2022byzantine} proposes to extend centered clipping over a distributed network to \textit{self centered clipping (SCC)} over a decentralized network. The aggregation rule of honest worker $n \in \mathcal{N}$ is given by
	\begin{align}
		&\A_n (\vx_n, \{\tilde{\vx}_{m,n}\}_{m\in\N_n\cup\B_n}) \\
		=& \sum_{m\in \N_n\cup\B_n \cup \{n\}} w^{\prime}_{nm} (\vx_n + \operatorname{CLIP}(\tilde{\vx}_{m,n}-\vx_n, \tau_n)),
		\nonumber \\
		=& (1-w^{\prime}_{nn}) \lp\vx_n +\!\!\!\!\!\! \sum_{m\in \N_n\cup\B_n} \frac{w^{\prime}_{nm}}{1-w^{\prime}_{nn}} \operatorname{CLIP}(\tilde{\vx}_{m,n}-\vx_n, \tau_n)\rp
		\nonumber \\
		& + w^{\prime}_{nn}\vx_n, \nonumber
	\end{align}
	where $w^{\prime}_{nm}$ is the weight that worker $n$ assigns to worker $m$ (see \eqref{rule:DECENTRALIZED SGD} for reference), $\tau_n > 0$ is a threshold, and $\operatorname{CLIP}$ is a function defined as
	\begin{align}
		\label{eq:clip}
		\operatorname{CLIP}(\vx, \tau) := \min\lp 1, \frac{\tau}{\|\vx\|}\rp\vx.
	\end{align}
	Therefore, SCC defines the base aggregator as
	\begin{align}
		&\A(\vx_n, \{\tilde{\vx}_{m,n}\}_{m\in\N_n\cup\B_n})    \\
		=& \vx_n +\!\!\!\!\!\! \sum_{m\in \N_n\cup\B_n} \frac{w^{\prime}_{nm}}{1-w^{\prime}_{nn}} \operatorname{CLIP}(\vx_m-\vx_n, \tau_n),
		\nonumber
	\end{align}
	and sets $r_n=w^{\prime}_{nn}$.
	
	The work of \cite{guo2021byzantine} defines a two-stage method \textit{UBAR} to filter out the potential Byzantine attacks. The base aggregator $\A$ is \textit{UBAR}, and the parameter $r_n$ is tunable.
	
	
	
	The work of \cite{xu2022byzantine} proposes a \emph{similarity-based reweighting (SimRew)} method. Worker $n$ assigns a weight $c_{nm}>0$ to worker $m\in\N\cup\B$ and computes an auxiliary vector $\vy_{nm}$ based on the similarity between $\vx_n$ and $\tilde{\vx}_{m,n}$. SimRew defines the base aggregator as
	\begin{align}
		\A_n (\vx_n, \{\tilde{\vx}_{m,n}\}_{m\in\N_n\cup\B_n})
		=& \vx_n+\sum_{m\in \N_n\cup\B_n} c_{nm} \vy_{nm},
		\nonumber \\
		=& \sum_{m\in \N_n\cup\B_n} c_{nm} (\vy_{nm}+\vx_n),
	\end{align}
	and sets $r_n=0$.
	
	
	\section{Useful Lemmas and Their Proofs}
	In this section we introduce some useful lemmas which are necessary in the proofs of main theorems.
	
	
	\subsection{Lemma \ref{lemma:submultiplicative} and Its Proof}
	
	The following lemma claims that the Frobenius norm $\|\cdot\|_{F}$ is sub-multiplicative.
	\begin{Lemma}[Sub-multiplicativity of $\|\cdot\|_{F}$]
		\label{lemma:submultiplicative}
		For the Frobenius norm $\|\cdot\|_{F}$, it holds for any $A,Z\in\R^{N\times N}$ that
		\begin{align}
			\|AZ\|_{F} \le \|A\| \|Z\|_{F},
		\end{align}
		where $\|\cdot\|$ is the spectral norm.
	\end{Lemma}
	\begin{proof}
		Decompose $Z$ by columns as $Z=[\vz_1, \cdots, \vz_N]$. Then, $AZ=[A\vz_1, \cdots, A\vz_N]$. It follows that
		\begin{align}
			\hspace{-1em} \!\!\|A Z\|_F^2
			\!=\!\sum_{n=1}^N\left\|A \vz_n\right\|^2
			\!\leq\!\|A\|^2 \sum_{n=1}^N\left\|\vz_n\right\|^2
			\!=\!\|A\|^2\|Z\|_F^2,
		\end{align}
		where $\|\cdot\|$ is the spectrum norm for a matrix and $\ell_2$ norm for a vector.
		This completes the proof.
	\end{proof}

	\subsection{Lemma \ref{lemma:ce-1/2-iteration} and Its Proof}
	The following Lemma characterizes the connection between $H^{k+\frac{1}{2}}$ and $H^{k}$ in the Byzantine-resilient decentralized SGD.
	
	
	\begin{Lemma}
		\label{lemma:ce-1/2-iteration}
		Consider the Byzantine-resilient decentralized SGD in Algorithm \ref{algorithm:ByrdDec}. Under Assumptions \ref{assumption:Lip}--\ref{assumption:outerVariance}, if a constant step size $\alpha^k=\alpha$ is used, for any $v \in (0, 1)$, it holds that
		\begin{align}
			\label{inequality:ce-1/2-iteration-I}
			\E_{\vxi^k}[H^{k+\frac{1}{2}}]
			\le& \lp\frac{1}{1-v}+\frac{6\alpha^2 L^2}{v}\rp H^k
			+\frac{4\alpha^2}{v}(\delta^2_{\rm in} +\delta^2_{\rm out}).
		\end{align}
	\end{Lemma}
	
	\begin{proof}
		Observe that
		\begin{align}
			\label{inequality:ce-1/2-iteration-1}
			& \E_{\vxi^k}[H^{k+\frac{1}{2}}]
			=\E_{\vxi^k}[\frac{1}{N}\sum_{n\in\N} \|\vx^{k+\frac{1}{2}}_{n}-\bar\vx^{k+\frac{1}{2}}\|^2]
			\\
			&\le \frac{1}{1-v}\frac{1}{N}\sum_{n\in\N} \| \vx^{k}_{n}-\bar\vx^{k}\|^2 \nonumber \\
			&+\frac{\alpha^2}{v}\E_{\vxi^k}[\frac{1}{N}\sum_{n\in\N} \| \nabla f_n(\vx^{k}_n; \xi_n^{k}) \!-\! \frac{1}{N}\sum_{m\in\N} \nabla f_m(\vx^{k}_m; \xi_m^{k}) \|^2].
			\nonumber
		\end{align}
		Applying variance decomposition $\E[\|\vx\|^2]=\|\E[\vx]\|^2+\E[\|\vx-\E[\vx]\|^2]$ to the second term at the RHS of \eqref{inequality:ce-1/2-iteration-1} yields
		\begin{align}
			\label{inequality:ce-1/2-iteration-1-2}
			& \hspace{-1em} \E_{\vxi^k}[\frac{1}{N}\sum_{n\in\N} \|\nabla f_n(\vx^{k}_n; \xi_n^{k})-\frac{1}{N}\sum_{m\in\N} \nabla f_m(\vx^{k}_m; \xi_m^{k})\|^2]
			\\
			&= \frac{1}{N}\sum_{n\in\N} \| \nabla f_n(\vx^{k}_n)-\frac{1}{N}\sum_{m\in\N} \nabla f_m(\vx^{k}_m)\|^2
			\nonumber \\
			&+ \E_{\vxi^k}[\frac{1}{N}\sum_{n\in\N} \| \nabla f_n(\vx^{k}_n; \xi_n^{k})- \frac{1}{N}\sum_{m\in\N} \nabla f_m(\vx^{k}_m; \xi_m^{k})
			\nonumber \\
			&\quad\quad\quad\quad\quad\quad\quad- (\nabla f_n(\vx^{k}_n) -\frac{1}{N}\sum_{m\in\N} \nabla f_m(\vx^{k}_m)) \|^2]. \nonumber
		\end{align}
		
		With Assumptions \ref{assumption:Lip} and \ref{assumption:outerVariance}, the first term at the RHS of \eqref{inequality:ce-1/2-iteration-1-2} can be bounded by
		\begin{align}
			\label{inequality:ce-1/2-iteration-1-2-1}
			&~~~\frac{1}{N}\sum_{n\in\N} \| \nabla f_n(\vx^{k}_n)-\frac{1}{N}\sum_{m\in\N} \nabla f_m(\vx^{k}_m)\|^2
			\\
			& \le 3\frac{1}{N}\sum_{n\in\N} \| \nabla f_n(\vx^{k}_n)-\nabla f_n(\bar\vx^{k})\|^2
			\nonumber \\
			& + 3\frac{1}{N}\sum_{n\in\N} \| \nabla f_n(\bar\vx^{k})-\frac{1}{N}\sum_{m\in\N} \nabla f_m(\bar\vx^{k})\|^2
			\nonumber \\
			& + 3\frac{1}{N}\sum_{n\in\N} \| \frac{1}{N}\sum_{m\in\N} \nabla f_m(\bar\vx^{k})-\frac{1}{N}\sum_{m\in\N} \nabla f_m(\vx^{k}_m)\|^2
			\nonumber \\
			&\le 6L^2 \frac{1}{N}\sum_{n\in\N} \|\vx^k_n-\bar\vx^{k}\|^2 + 3\delta^2_{\rm out}. \nonumber
		\end{align}
		With Assumption \ref{assumption:innerVariance}, the second term at the RHS of \eqref{inequality:ce-1/2-iteration-1-2} can be bounded by
		\begin{align}
			\label{inequality:ce-1/2-iteration-1-2-2}
			&\E_{\vxi^k}[\frac{1}{N}\sum_{n\in\N} \| \nabla f_n(\vx^{k}_n; \xi_n^{k})\!-\! \frac{1}{N}\sum_{m\in\N} \nabla f_m(\vx^{k}_m; \xi_m^{k})
			\\
			& \quad\quad\quad\quad\quad\quad- (\nabla f_n(\vx^{k}_n) -\frac{1}{N}\sum_{m\in\N}  \nabla f_m(\vx^{k}_m)) \|^2]
			\nonumber \\
			\le & 2\frac{1}{N}\sum_{n\in\N} \E_{\vxi^k}[\|\nabla f_n(\vx^{k}_n; \xi_n^{k})- \nabla f_n(\vx^{k}_n) \|^2]
			\nonumber \\
			+ & 2\frac{1}{N}\sum_{n\in\N} \E_{\vxi^k}[\|\frac{1}{N}\!\sum_{m\in\N} \!\!\nabla f_m(\vx^{k}_m; \xi_m^{k}) \!-\!\!\frac{1}{N}\!\sum_{m\in\N} \!\!\nabla f_m(\vx^{k}_m)\|^2]
			\nonumber \\
			\le & 4\delta^2_{\rm in}. \nonumber
		\end{align}
		Substituting \eqref{inequality:ce-1/2-iteration-1-2-1} and \eqref{inequality:ce-1/2-iteration-1-2-2} into \eqref{inequality:ce-1/2-iteration-1-2} yields \eqref{inequality:ce-1/2-iteration-I}.
	\end{proof}

	\subsection{Lemma \ref{lemma:different-partial-average} and Its Proof}
	Lemma \ref{lemma:different-partial-average} describes the difference between partial weighted average and full weighted average for a set of vectors.
	\begin{Lemma}
		\label{lemma:different-partial-average}
		Consider the two sets $\N_1 \subseteq \N_2$ and define their weighted means as
		\begin{align}
			\bar\vy_1 := \frac{\sum_{n\in\N_1} c_n\vx_n}{\sum_{n\in\N_1} c_n},
			~~~
			\bar\vy_2 := \frac{\sum_{n\in\N_2} c_n\vx_n}{\sum_{n\in\N_2} c_n},
		\end{align}
		in which $\{c_n \}_{n\in\N_2}$ are a set of positive weights. The difference between the weighted means is bounded by
		\begin{align}
			\label{inequality:different-partial-average-1}
			\lnorm \bar\vy_1-\bar\vy_2 \rnorm \le& \frac{\sum_{n\in\N_2\setminus\N_1} c_n}{\sum_{n\in\N_2} c_n}\max_{n\in\N_2\setminus\N_1}\lnorm \vx_n - \bar\vy_1\rnorm,
			\\
			\label{inequality:different-partial-average-2}
			\lnorm \bar\vy_1-\bar\vy_2 \rnorm \le& \frac{\sum_{n\in\N_2\setminus\N_1} c_n}{\sum_{n\in\N_1} c_n}\max_{n\in\N_2\setminus\N_1}\lnorm \vx_n - \bar\vy_2\rnorm.
		\end{align}
	\end{Lemma}
	\begin{proof}
		Expanding the difference $\|\bar\vy_1-\bar\vy_2\|$ yields
		\begin{align}
			&\|\bar\vy_1-\bar\vy_2\|
			\\
			=& \lnorm\frac{\sum_{n\in\N_1} c_n\vx_n}{\sum_{n\in\N_1} c_n}-\frac{\sum_{n\in\N_2} c_n\vx_n}{\sum_{n\in\N_2} c_n}\rnorm
			\nonumber \\
			=& \left\|\lp\frac{1}{\sum_{n\in\N_1} c_n}-\frac{1}{\sum_{n\in\N_2} c_n}\rp\sum_{n\in\N_1} c_n\vx_n \right.
			\nonumber \\
			&\left. \hspace{0.2cm}
			- \frac{\sum_{n\in\N_2\setminus\N_1} c_n\vx_n}{\sum_{n\in\N_2} c_n}\right\|
			\nonumber \\
			=& \frac{\sum_{n\in\N_2\setminus\N_1} c_n}{\sum_{n\in\N_2} c_n} \lnorm\frac{\sum_{n\in\N_1} c_n\vx_n}{\sum_{n\in\N_1} c_n} - \frac{\sum_{n\in\N_2\setminus\N_1} c_n\vx_n}{\sum_{n\in\N_2\setminus\N_1} c_n}\rnorm
			\nonumber \\
			=& \frac{\sum_{n\in\N_2\setminus\N_1} c_n}{\sum_{n\in\N_2} c_n} \lnorm\bar\vy_1 - \frac{\sum_{n\in\N_2\setminus\N_1} c_n\vx_n}{\sum_{n\in\N_2\setminus\N_1} c_n}\rnorm
			\nonumber \\
			\le& \frac{\sum_{n\in\N_2\setminus\N_1} c_n}{\sum_{n\in\N_2} c_n}\max_{n\in\N_2\setminus\N_1}\lnorm \vx_n - \bar\vy_1\rnorm, \nonumber
		\end{align}
		where the last step comes from applying Jensen's inequality. This completes the proof of \eqref{inequality:different-partial-average-1}. The proof of \eqref{inequality:different-partial-average-2} is similar and we omit it for simplicity.
	\end{proof}

	\ifinarxiv{
	}{
	}
	
	\begin{table*}[t]
		\small
		\centering
		\caption{Accuracy (Acc) and disagreement measure (DM) in the Erdos-Renyi graph for the i.i.d. case.}
		\label{table:er-iid}
		\setlength{\tabcolsep}{1.0mm}{
			\begin{tabular}{c|cc|cc|cc|cc|cc|cc}
\hline\hline
\multirow{2}{*}{}&\multicolumn{2}{c|}{no attack}&\multicolumn{2}{c|}{Gaussian}&\multicolumn{2}{c|}{sign-flipping}&\multicolumn{2}{c|}{isolation}&\multicolumn{2}{c|}{sample-duplicating}&\multicolumn{2}{c}{ALIE}\\
& Acc.(\%) & CE & Acc.(\%) & CE & Acc.(\%) & CE & Acc.(\%) & CE & Acc.(\%) & CE & Acc.(\%) & CE \\
\hline
no comm. & 90.24 & $>$1e-01 & -- & --  & -- & --  & -- & --  & -- & --  & -- & -- \\
\hline
WeiMean & {91.70} & {$<$1e-07} & {13.78} & {$>$1e-01} & {30.93} & {3e-07} & {90.18} & {$>$1e-01} & {91.63} & {$<$1e-07} & {91.67} & {$<$1e-07}\\
CooMed & {91.58} & {$<$1e-07} & {91.56} & {1e-07} & {87.00} & {1e-07} & {91.52} & {2e-07} & {91.46} & {1e-07} & {91.52} & {$<$1e-07}\\
GeoMed & {91.70} & {$<$1e-07} & {91.68} & {$<$1e-07} & {87.10} & {$<$1e-07} & {91.29} & {$>$1e-01} & {91.70} & {$<$1e-07} & {91.69} & {$<$1e-07}\\
Krum & {90.83} & {2e-07} & {90.83} & {3e-07} & {91.23} & {3e-07} & {90.89} & {5e-07} & {90.91} & {7e-07} & {91.45} & {$<$1e-07}\\
TriMean & \textbf{91.73} & {$<$1e-07} & {91.61} & {$<$1e-07} & {86.34} & {$<$1e-07} & {91.57} & {1e-07} & {91.61} & {$<$1e-07} & {91.62} & {$<$1e-07}\\
SimRew & {79.15} & {$>$1e-01} & {76.88} & {$>$1e-01} & {76.87} & {$>$1e-01} & {76.93} & {$>$1e-01} & {76.86} & {$>$1e-01} & {76.87} & {$>$1e-01}\\
DRSA & {91.69} & {4e-06} & {91.62} & {4e-06} & {90.27} & {7e-06} & {91.61} & {7e-03} & {91.64} & {3e-06} & {91.64} & {3e-06}\\
CC & {91.68} & {$<$1e-07} & {91.69} & {2e-06} & {31.57} & {3e-07} & {91.51} & {7e-06} & {91.63} & {$<$1e-07} & {91.63} & {$<$1e-07}\\
SCC & {91.72} & {$<$1e-07} & \textbf{91.74} & {2e-06} & {32.56} & {3e-07} & {91.59} & {8e-06} & \textbf{91.71} & {$<$1e-07} & \textbf{91.70} & {$<$1e-07}\\
FABA & {91.69} & {$<$1e-07} & {91.67} & {$<$1e-07} & \textbf{91.71} & {$<$1e-07} & \textbf{91.68} & {$<$1e-07} & {91.64} & {$<$1e-07} & {91.67} & {$<$1e-07}\\
\textbf{IOS (ours)} & {91.68} & {$<$1e-07} & {91.68} & {$<$1e-07} & {91.66} & {$<$1e-07} & {91.66} & {$<$1e-07} & {91.70} & {$<$1e-07} & {91.68} & {$<$1e-07}\\
\hline\hline
\end{tabular}

		}
		
		\vspace{5mm}
		\small
		\centering
		\caption{Accuracy (Acc) and disagreement measure (DM) in the Erdos-Renyi graph for the non-i.i.d. case.}
		\label{table:er-non-iid}
		\setlength{\tabcolsep}{1.0mm}{
			\begin{tabular}{c|cc|cc|cc|cc|cc|cc}
\hline\hline
\multirow{2}{*}{}&\multicolumn{2}{c|}{no attack}&\multicolumn{2}{c|}{Gaussian}&\multicolumn{2}{c|}{sign-flipping}&\multicolumn{2}{c|}{isolation}&\multicolumn{2}{c|}{sample-duplicating}&\multicolumn{2}{c}{ALIE}\\
& Acc.(\%) & CE & Acc.(\%) & CE & Acc.(\%) & CE & Acc.(\%) & CE & Acc.(\%) & CE & Acc.(\%) & CE \\
\hline
no comm. & 10.00 & $>$1e-01 & -- & --  & -- & --  & -- & --  & -- & --  & -- & -- \\
\hline
WeiMean & \textbf{91.70} & {9e-07} & {13.63} & {$>$1e-01} & {20.89} & {5e-06} & {10.00} & {$>$1e-01} & \textbf{91.53} & {9e-07} & \textbf{91.21} & {7e-07}\\
CooMed & {75.49} & {2e-02} & {77.92} & {5e-07} & {66.38} & {5e-03} & {41.56} & {$>$1e-01} & {78.74} & {4e-07} & {82.46} & {3e-07}\\
GeoMed & {88.00} & {2e-07} & {88.63} & {3e-07} & {23.35} & {8e-07} & {73.01} & {$>$1e-01} & {88.68} & {6e-07} & {89.14} & {2e-07}\\
Krum & {19.71} & {4e-05} & {19.76} & {4e-05} & {19.16} & {2e-07} & {19.76} & {4e-05} & {29.20} & {3e-05} & {63.79} & {3e-06}\\
TriMean & {91.61} & {5e-07} & {90.40} & {2e-06} & {70.97} & {4e-06} & {54.47} & {6e-06} & {87.54} & {3e-06} & {82.56} & {2e-06}\\
SimRew & {10.40} & {$>$1e-01} & {10.53} & {$>$1e-01} & {10.53} & {$>$1e-01} & {10.53} & {$>$1e-01} & {10.53} & {$>$1e-01} & {10.53} & {$>$1e-01}\\
DRSA & {76.43} & {3e-03} & {89.38} & {3e-03} & {10.50} & {4e-03} & {88.60} & {1e-03} & {85.27} & {4e-03} & {60.69} & {4e-03}\\
CC & {91.60} & {5e-07} & {91.56} & {2e-05} & {45.16} & {3e-06} & {80.03} & {8e-05} & \textbf{91.53} & {7e-07} & {91.13} & {5e-07}\\
SCC & \textbf{91.70} & {9e-07} & {91.57} & {2e-05} & {20.89} & {5e-06} & {83.80} & {8e-05} & {91.52} & {9e-07} & \textbf{91.21} & {7e-07}\\
FABA & {91.61} & {5e-07} & {91.54} & {5e-07} & {91.55} & {5e-07} & {91.53} & {5e-07} & {88.91} & {1e-06} & {86.87} & {9e-07}\\
\textbf{IOS (ours)} & \textbf{91.70} & {9e-07} & \textbf{91.60} & {8e-07} & \textbf{91.61} & {8e-07} & \textbf{91.55} & {8e-07} & {88.22} & {2e-06} & {86.61} & {2e-06}\\
\hline\hline
\end{tabular}

		}
	\end{table*}
	
	\section{Additional Numerical Experiments}
	\label{section:supplementary-experiments}
	In addition to the elaborately crafted two-castle and octopus graphs, we also have additional numerical experiments in an Erdos-Renyi graph with $10$ honest workers and $2$ Byzantine workers. Each pair of workers are neighbors with the probability of $0.7$. The results can be found in Fig. \ref{fig:covergence}, as well as Tables 
	\ref{table:er-iid} and \ref{table:er-non-iid}.
	The observations are consistent with those in the two-castle and octopus graphs, illustrating the state-of-the-art performance of the proposed IOS.

\end{document}